\documentclass[reqno]{amsart}

\usepackage{booktabs} 
\usepackage{IEEEtrantools}
\usepackage{amssymb,latexsym,amsfonts,amsmath}
\usepackage{graphicx}
\usepackage{dsfont}

\usepackage{algorithm}
\usepackage{algorithmic}

\newtheorem{theorem}{Theorem}[section]

\newtheorem{corollary}[theorem]{Corollary}

\newtheorem{definition}[theorem]{Definition}

\newtheorem{remark}[theorem]{Remark}
\newtheorem{assumption}{Assumption}
\numberwithin{equation}{section}

\newcommand{\R}{{\mathbb{R}}}

\newcommand{\N}{{\mathbb{N}}}

\newcommand{\Let}{:=}

\begin{document}

\begin{abstract}
This paper is concerned with a compositional approach for constructing finite Markov decision processes of interconnected discrete-time stochastic control systems. The proposed approach leverages the interconnection topology and a notion of so-called \emph{stochastic storage functions} describing joint dissipativity-type properties of subsystems and their abstractions.
In the first part of the paper, we derive dissipativity-type compositional conditions for quantifying the error between the interconnection of stochastic control subsystems and that of their abstractions. In the second part of the paper, we propose an approach to construct finite Markov decision processes together with their corresponding stochastic storage functions for classes of discrete-time control systems satisfying some \emph{incremental passivablity} property. Under this property, one can construct finite Markov decision processes by a suitable discretization of the input and state sets. Moreover, we show that for linear stochastic control systems, the aforementioned property can be readily checked by some matrix inequality. We apply our proposed results to the temperature regulation in a circular building by constructing compositionally a finite Markov decision process of a network containing $200$ rooms in which the compositionality condition does not require any constraint on the number or gains of the subsystems. We employ the constructed finite Markov decision process as a substitute to synthesize policies regulating the temperature in each room for a bounded time horizon.
\end{abstract}

\title[From Dissipativity Theory to Compositional Construction of Finite Markov Decision Processes]{From Dissipativity Theory to Compositional Construction of Finite Markov Decision Processes}

\author{Abolfazl Lavaei$^1$}
\author{Sadegh Soudjani$^2$}
\author{Majid Zamani$^1$}
\address{$^1$Hybrid Control Systems Group, Technical University of Munich, Germany.}
\email{lavaei@tum.de, zamani@tum.de}
\address{$^2$School of Computing, Newcastle University, UK.}
\email{sadegh.soudjani@ncl.ac.uk}
\maketitle

\section{Introduction}

Large-scale interconnected systems have received significant attentions in the last few years due to their presence in real life systems including power networks, air traffic control, and so on. Each complex real-world system can be regarded as an interconnected system composed of several subsystems.
Since these large-scale networks of systems are inherently difficult to analyze and control, one can develop compositional schemes to employ the abstractions of the given subsystems as a replacement in the controller design process. Those abstractions allow us to design controllers for them, and then refine the controllers to the ones for the concrete subsystems, while provide us with the quantified errors for the overall interconnected system in this controller synthesis detour.

Construction of finite abstractions was introduced in recent years as a method to reduce the complexity of controller synthesis problems in particular for enforcing complex logical properties. Finite abstractions are abstract descriptions of the continuous-space control systems in which each discrete state corresponds to a collection of continuous states of the original system. Since the abstractions are finite, algorithmic approaches from computer science are applicable to synthesize controllers enforcing complex logic properties including those expressed as linear temporal logic formulae.

In the past few years, there have been several results on the construction of (in)finite abstractions for stochastic systems.
Existing results for \emph{continuous-time} systems include infinite approximation techniques for jump-diffusion systems \cite{julius2009approximations}, finite bisimilar abstractions for incrementally stable stochastic switched systems \cite{zamani2015symbolic} and randomly switched stochastic systems \cite{zamani2014approximately}, and finite bisimilar abstractions for incrementally stable stochastic control systems without discrete dynamics \cite{zamani2014symbolic}.
Recently, compositional construction of infinite abstractions is discussed in \cite{zamani2016approximations} using small-gain type conditions
and of finite bisimilar abstractions in \cite{2017arXiv170909546M} based on a new notion of disturbance bisimilarity relation.

For \emph{discrete-time} stochastic models with continuous state spaces, finite approximations are initially proposed in \cite{APLS08} for formal verification and synthesis of this class of systems. The algorithms are improved in terms of scalability in \cite{SA13,SSoudjani}. Those techniques have been implemented in the tool \texttt{FAUST} \cite{FAUST15}. Extension of the techniques to infinite horizon properties is proposed in \cite{tkachev2011infinite} and formal abstraction-based policy synthesis is discussed in \cite{tmka2013}. Recently, compositional construction of finite abstractions is discussed in \cite{SAM15} using dynamic Bayesian networks and infinite abstractions in \cite{lavaei2017compositional} using small-gain type conditions both for discrete-time stochastic control systems. Our proposed approach extends the abstraction techniques in \cite{SAM15} from verification to synthesis, by proposing a different quantification of the abstraction error, and leveraging the dissipativity properties of subsystems and structure of interconnection topology to show the compositonal results for the finite Markov decision processes.
Although the results in \cite{lavaei2017compositional} deal only with infinite abstractions (reduced order models), our proposed approach considers finite Markov decision processes as abstractions which are the main tools for automated synthesis of controllers for complex logical properties. To the best of our knowledge, this is the first time a closed form dynamical representation of the abstract finite Markov decision processes is used to facilitate the use of dissipativity properties of subsystems in the error quantification.
%


In particular, we provide a compositional approach for the construction of finite Markov decision processes of interconnected discrete-time stochastic control systems. The proposed compositional technique leverages the interconnection structure and joint dissipativity-type properties of subsystems and their abstractions characterized via a notion of so-called \emph{stochastic storage functions}. The provided compositionality conditions can enjoy the structure of interconnection topology and be potentially satisfied independently of the number or gains of the subsystems (cf. case study section). The stochastic storage functions of subsystems are utilized to quantify the error in probability between the interconnection of concrete stochastic subsystems and that of their finite Markov decision processes. As a consequence, one can leverage the proposed results here to solve particularly safety/reachability problems over the finite interconnected systems and then carry the results over the concrete interconnected ones.

We also propose an approach to construct finite Markov decision processes together with their corresponding stochastic storage functions for classes of stochastic control subsystems satisfying some \emph{incremental passivability} property. Under this property, one can construct a finite Markov decision process by a suitable discretization of the input and state sets. Moreover, we show that for linear stochastic control systems, the mentioned property can be readily verified by some matrix inequality. Finally, we illustrate the effectiveness of the results using the temperature regulation in a circular building by constructing compositionally a finite Markov decision process of a network containing $200$ rooms in which the compositionality condition does not require any constraint on the number or gains of the subsystems. We leverage the constructed finite Markov decision process as a substitute to synthesize policies regulating the temperature in each room for a bounded time horizon. We benchmark our results against the compositional abstraction technique of \cite{SAM15} which is based on construction of finite dynamic Bayesian networks. 


\section{Discrete-Time Stochastic Control Systems}

\subsection{Preliminaries}
We consider a probability space $(\Omega,\mathcal F_{\Omega},\mathbb{P}_{\Omega})$,
where $\Omega$ is the sample space,
$\mathcal F_{\Omega}$ is a sigma-algebra on $\Omega$ comprising subsets of $\Omega$ as events,
and $\mathbb{P}_{\Omega}$ is a probability measure that assigns probabilities to events.
We assume that random variables introduced in this article are measurable functions of the form $X:(\Omega,\mathcal F_{\Omega})\rightarrow (S_X,\mathcal F_X)$.
Any random variable $X$ induces a probability measure on  its space $(S_X,\mathcal F_X)$ as $Prob\{A\} = \mathbb{P}_{\Omega}\{X^{-1}(A)\}$ for any $A\in \mathcal F_X$.
We often directly discuss the probability measure on $(S_X,\mathcal F_X)$ without explicitly mentioning the underlying probability space and the function $X$ itself.

A topological space $S$ is called a Borel space if it is homeomorphic to a Borel subset of a Polish space (i.e., a separable and completely metrizable space).
Examples of a Borel space are the Euclidean spaces $\mathbb R^n$, its Borel subsets endowed with a subspace topology, as well as hybrid spaces.
Any Borel space $S$ is assumed to be endowed with a Borel sigma-algebra, which is
denoted by $\mathcal B(S)$. We say that a map $f : S\rightarrow Y$ is measurable whenever it is Borel measurable.


\subsection{Notation}

The following notation is used throughout the paper. We denote the set of nonnegative integers by $\mathbb N := \{0,1,2,\ldots\}$ and the set of positive integers by $\mathbb N_{\ge 1} := \{1,2,3,\ldots\}$. 
The symbols $\R$, $\R_{>0}$, and $\R_{\ge 0}$ denote the set of real, positive and nonnegative real numbers, respectively.
For any set $X$ we denote by $2^X$ the power set of $X$ that is the set of all subsets of $X$.
Given $N$ vectors $x_i \in \R^{n_i}$, $n_i\in \mathbb N_{\ge 1}$, and $i\in\{1,\ldots,N\}$, we use $x = [x_1;\ldots;x_N]$ to denote the corresponding vector of dimension $\sum_i n_i$.
Given a vector $x\in\mathbb{R}^{n}$, $\Vert x\Vert$ denotes the Euclidean norm of $x$. The symbol $I_n$ denotes the identity matrix in $\R^{n\times{n}}$. Also, the identity map on a set $A$ in denoted by $\mathbf{1}_A$. We denote by $\mathsf{diag}(a_1,\ldots,a_N)$ a diagonal matrix in $\R^{N\times{N}}$ with diagonal matrix entries $a_1,\ldots,a_N$ starting from the upper left corner. Given functions $f_i:X_i\rightarrow Y_i$,
for any $i\in\{1,\ldots,N\}$, their Cartesian product $\prod_{i=1}^{N}f_i:\prod_{i=1}^{N}X_i\rightarrow\prod_{i=1}^{N}Y_i$ is defined as $(\prod_{i=1}^{N}f_i)(x_1,\ldots,x_N)=[f_1(x_1);\ldots;f_N(x_N)]$.
For any set $A$ we denote by $A^{\mathbb N}$ the Cartesian product of a countable number of copies of $A$, i.e., $A^{\mathbb N} = \prod_{k=0}^{\infty} A$.
Given a measurable function $f:\mathbb N\rightarrow\mathbb{R}^n$, the (essential) supremum of $f$ is denoted by $\Vert f\Vert_{\infty} \Let \text{(ess)sup}\{\Vert f(k)\Vert,k\geq 0\}$. A function $\gamma:\mathbb\R_{0}^{+}\rightarrow\mathbb\R_{0}^{+}$, is said to be a class $\mathcal{K}$ function if it is continuous, strictly increasing, and $\gamma(0)=0$. A class $\mathcal{K}$ function $\gamma(\cdot)$ is said to be a class $\mathcal{K}_{\infty}$ if
$\lim_{r\rightarrow\infty}\gamma(r) = \infty$.

\subsection{Discrete-Time Stochastic Control Systems}
We consider stochastic control systems in discrete time (dt-SCS) defined over a general state space
and characterized by the tuple
\begin{equation}
\label{eq:dt-SCS}
\Sigma\!=\!\left(X,U,W,\varsigma,f,Y_1,Y_2, h_1, h_2\right)\!,
\end{equation}
where $X$ is a Borel space as the state space of the system.
We denote by $(X, \mathcal B (X))$ the measurable space
with $\mathcal B (X)$  being  the Borel sigma-algebra on the state space. Sets
$U$ and $W$ are Borel spaces as the \emph{external} and \emph{internal} input spaces of the system.
Notation $\varsigma$ denotes a sequence of independent and identically distributed (i.i.d.) random variables on a set $V_\varsigma$
\begin{equation*}
\varsigma:=\{\varsigma(k):\Omega\rightarrow V_{\varsigma},\,\,k\in\N\}.
\end{equation*}
The map $f:X\times U\times W\times V_{\varsigma} \rightarrow X$ is a measurable function characterizing the state evolution of the system.	
Finally, sets $Y_1$ and $Y_2$ are Borel spaces as the external and internal output spaces of the system, respectively.
Maps $h_1:X\rightarrow Y_1$ and $h_2:X\rightarrow Y_2$ are measurable functions that map a state $x\in X$ to its external and internal outputs $y_1 = h_1(x)$ and $y_2 = h_2(x)$, respectively.

For given initial state $x(0)\in X$ and input sequences $\nu(\cdot):\mathbb N\rightarrow U$ and $w(\cdot):\mathbb N\rightarrow W$, evolution of the state of dt-SCS $\Sigma$ can be written as
\begin{equation}\label{Eq_1a}
\Sigma:\left\{\hspace{-1.5mm}\begin{array}{l}x(k+1)=f(x(k),\nu(k),w(k),\varsigma(k))\\
y_1(k)=h_1(x(k))\\
y_2(k)=h_2(x(k)),\\
\end{array}\right.
\quad k\in\mathbb N.
\end{equation}


Given the dt-SCS in \eqref{eq:dt-SCS}, we are interested in \emph{Markov policies} to control the system.
\begin{definition}
	A Markov policy for the dt-SCS $\Sigma$ in \eqref{eq:dt-SCS} is a sequence
	$\rho = (\rho_0,\rho_1,\rho_2,\ldots)$ of universally measurable stochastic kernels $\rho_n$ \cite{BS96},
	each defined on the input space $U$ given $X\times W$ and such that for all $(x_n,w_n)\in X\times W$, $\rho_n(U(x_n,w_n)|(x_n,w_n))=1$.
	The class of all such Markov policies is denoted by $\Pi_M$. 
\end{definition} 

%
We associate respectively to $U$ and $W$ the sets $\mathcal U$ and $\mathcal W$ to be collections of sequences $\{\nu(k):\Omega\rightarrow U,\,\,k\in\N\}$ and $\{w(k):\Omega\rightarrow W,\,\,k\in\N\}$, in which $\nu(k)$ and $w(k)$ are independent of $\varsigma(t)$ for any $k,t\in\mathbb N$ and $t\ge k$.
For any initial state $a\in X$, $\nu(\cdot)\in\mathcal{U}$, and $w(\cdot)\in\mathcal{W}$,
the random sequences $x_{a\nu w}:\Omega \times\N \rightarrow X$, $y^1_{a\nu w}:\Omega \times \N \rightarrow Y_1$ and $y^2_{a\nu w}:\Omega \times \N \rightarrow Y_2$ that satisfy \eqref{Eq_1a}
are called respectively the \textit{solution process} and external and internal \textit{output trajectory} of $\Sigma$ under external input $\nu$, internal input $w$ and initial state $a$.


In this sequel we assume that the state space $X$ of $\Sigma$ is a subset of $\mathbb R^n$. System $\Sigma$ is called finite if $ X, U, W$ are finite sets and infinite otherwise.
In this paper we are interested in studying interconnected discrete-time stochastic control systems without internal inputs and outputs that result from the interconnection of dt-SCS having both internal and external inputs and outputs. In this case, the interconnected dt-SCS without internal input and output in indicated by the simplified tuple $(X,U,\varsigma,f,Y,h)$ with $f:X\times U\times V_\varsigma\rightarrow X$.

\subsection{General Markov Decision Processes}
\label{subsec:MDP}
A dt-SCS $\Sigma$ in \eqref{eq:dt-SCS} can be \emph{equivalently} represented as a general Markov decision process (gMDP) \cite{SIAM17}
\begin{equation}\notag
\Sigma\!=\!\left(X,W,U,T_{\mathsf x},Y_1,Y_2,h_1, h_2\right),	
\end{equation}
where the map $T_{\mathsf x}:\mathcal B(X)\times X\times U\times W\rightarrow[0,1]$,
is a conditional stochastic kernel that assigns to any $x \in X$, $w\in W$ and $\nu\in U$ a probability measure $T_{\mathsf x}(\cdot | x,\nu, w)$
on the measurable space
$(X,\mathcal B(X))$
so that for any set $A \in \mathcal B(X)$, 
$$\mathbb P (x(k+1)\in A\,|\, x(k),\nu(k),w(k)) = \int_A T_{\mathsf x} (d\bar x|x(k),\nu(k),w(k)).$$




%

For given inputs $\nu(\cdot), w(\cdot),$  the stochastic kernel $T_{\mathsf x}$ captures the evolution of the state of $\Sigma$ and can be uniquely determined by the pair $(\varsigma,f)$ from \eqref{eq:dt-SCS}.


The alternative representation as gMDP is utilized in \cite{SAM15} to approximate a dt-SCS $\Sigma$ with a \emph{finite} $\widehat\Sigma$. Algorithm~~\ref{algo:MC_app} adapted from \cite{SAM15} with some modifications presents this approximation. The algorithm first constructs a finite partition of state set $X$ and input sets $U$, $W$.
Then representative points $\bar x_i\in \mathsf X_i$, $\bar \nu_i\in \mathsf U_i$ and $\bar w_i\in \mathsf W_i$ are selected as abstract states and inputs.
Transition probabilities in the finite gMDP $\widehat\Sigma$ are also computed according to \eqref{eq:trans_prob}. The output maps $\hat h_1,\hat h_2$ are the same as $h_1,h_2$ with their domain restricted to finite state set $\hat X$ (cf. Step \ref{step:output_map}) and the output sets $\hat Y_1,\hat Y_2$ are just image of $\hat X$ under $h_1,h_2$, respectively (cf. Step \ref{step:output_space}).

\begin{algorithm}[h]
	\caption{Abstraction of dt-SCS $\Sigma$ by a finite gMDP $\widehat\Sigma$}
	\label{algo:MC_app}
	\begin{center}
		\begin{algorithmic}[1]
			\REQUIRE 
			input dt-SCS $\Sigma\!=\!\left(X,W,U,T_{\mathsf x},Y_1,Y_2,h_1, h_2\right)$
			\STATE
			Select finite partitions of sets $X,U,W$ as $X = \cup_{i=1}^{n_x} \mathsf X_i$, $U = \cup_{i=1}^{n_\nu} \mathsf U_i$, $W = \cup_{i=1}^{n_w} \mathsf W_i$				
			\STATE
			For each $\mathsf X_i,\mathsf U_i$, and $\mathsf W_i$, select single representative points $x_i \in \mathsf X_i$, $\nu_i \in \mathsf U_i$, $w_i \in \mathsf W_i$
			\STATE
			Define 
			$\hat X := \{x_i, i=1,...,n_x\}$ as the finite state set of gMDP~$\widehat\Sigma$ with external and internal input sets
			$\hat U := \{\nu_i, i=1,...,n_\nu\}$ $\hat W := \{w_i, i=1,...,n_w\}$
			\STATE
			\label{step:refined}
			Define the map $\Xi:X\rightarrow 2^X$ that assigns to any $x\in X$, the corresponding partition set it belongs to, i.e.,
			$\Xi(x) = \mathsf X_i$ if $x\in \mathsf X_i$ for some $i=1,2,\ldots,n_x$
			\STATE
			Compute the discrete transition probability matrix $\hat T_{\mathsf x}$ for $\widehat\Sigma$ as:
			\begin{equation}
			\label{eq:trans_prob}
			\hat T_{\mathsf x} (x'|x,\nu,w) 
			= T_{\mathsf x} (\Xi(x')|x,\nu,w),
			\end{equation}
			for all $x,x'\in \hat X, \nu\in \hat U, w\in\hat W$
			\STATE
			\label{step:output_space}
			Define output spaces $\hat Y_1 := h_1(\hat X), \hat Y_2 := h_2(\hat X)$
			\STATE
			\label{step:output_map}
			Define output maps $\hat h_1 := h_1|_{\hat X}$ and $\hat h_2 := h_2|_{\hat X}$
			\ENSURE
			output finite gMDP $\widehat\Sigma = (\hat X, \hat U,\hat W, \hat T_{\mathsf x}, \hat Y_1, \hat Y_2, \hat h_1, \hat h_2)$
		\end{algorithmic}
	\end{center}
\end{algorithm}

In the following theorem we give a dynamical representation of the finite gMDP, which is more suitable for the study of this paper.
\begin{theorem}\label{Def154}
	Given a dt-SCS $\Sigma=\left(X,U,W,\varsigma,f,Y_1,Y_2, h_1, h_2\right)$,
	the finite gMDP $\widehat\Sigma$ constructed in Algorithm~\ref{algo:MC_app} can be represented as 		
	\begin{equation}
	\label{eq:abs_tuple}
	\hat\Sigma =(\hat X, \hat U,\hat W, \varsigma,\hat f,\hat Y_1,\hat Y_2,\hat h_1,\hat h_2),
	\end{equation}
	where $\hat f:\hat X\times\hat U\times\hat W\times V_\varsigma\rightarrow\hat X$ is defined as
	\begin{equation*}
	\hat f(\hat{x},\hat{\nu},\hat{w},\varsigma) = \Pi(f(\hat{x},\hat{\nu},\hat{w},\varsigma)),	
	\end{equation*}
	and $\Pi:X\rightarrow \hat X$ is the map that assigns to any $x\in X$, the representative point $\hat x\in\hat X$ of the corresponding partition set containing $x$.
	The initial state of $\widehat\Sigma$ is also selected according to $\hat x_0 := \Pi(x_0)$ with $x_0$ being the initial state of $\Sigma$. 
\end{theorem}
\begin{proof}
	It is sufficient to show that \eqref{eq:trans_prob} holds for dynamical representation of $\widehat\Sigma$ in \eqref{eq:abs_tuple} and that of $\Sigma$.
	For any $x,x'\in\hat X$, $\nu\in \hat U$ and $w\in\hat W$,
	\begin{align*}
	\hat T_{\mathsf x} (x'|x,\nu,w) & = \mathbb P(x' = \hat f(x,\nu,w,\varsigma)) = \mathbb P(x' = \Pi(f(x,\nu,w,\varsigma)))
	= \mathbb P(f(x,\nu,w,\varsigma)\in\Xi(x')),
	\end{align*}
	where $\Xi(x')$ is the partition set with $x'$ as its representative point as defined in Step~\ref{step:refined} of Algorithm~\ref{algo:MC_app}. Using the probability measure $\vartheta(\cdot)$ of random variable $\varsigma$ we can write
	\begin{align*}
	\hat T_{\mathsf x} (x'|x,\nu,w) = \int_{\Xi(x')}f(x,\nu,w,\varsigma)d\vartheta(\varsigma) = T_{\mathsf x} (\Xi(x')|x,\nu,w),
	\end{align*}
	which completes the proof.
\end{proof}
Dynamical representation provided by Theorem~\ref{Def154} uses the map $\Pi:X\rightarrow \hat X$ that assigns to any $x\in X$, the representative point $\hat x\in\hat X$ of the corresponding partition set containing $x$.
This map satisfies the inequality
\begin{equation}
\label{eq:Pi_delta}
\Vert \Pi(x)-x\Vert \leq \delta,\quad \forall x\in X,
\end{equation}
where $\delta:=\sup\{\|x-x'\|,\,\, x,x'\in \mathsf X_i,\,i=1,2,\ldots,n_x\}$ is the discretization parameter. We use this inequality in Section~\ref{sec:constrcution_finite} for compositional construction of finite gMDPs.

Algorithm~\ref{algo:MC_app} is used in \cite{SAM15} for \emph{compositional verification} of interconnected dt-SCS.
In order to provide formal guarantee on the compositional approximation, \cite{SAM15} uses distance in probability as a metric. In other words, for a given specification $\varphi$ and accuracy level $\epsilon$, the discretization parameters for each subsystem can be selected a priori such that after composition
\begin{equation}
\label{eq:metric_lit}
|\mathbb P(\Sigma\vDash\varphi) - \mathbb P(\widehat\Sigma\vDash\varphi)|\le \epsilon,
\end{equation}
where $\epsilon$ depends on the horizon of formula $\varphi$, Lipschitz constants of the stochastic kernels of subsystems, discretization parameters, and structure of the interconnection (cf. \cite[Theorem~9]{SAM15}).

In the next sections, we provide an approach for \emph{compositional synthesis} of interconnected dt-SCS. We first define the notions of stochastic storage and simulation functions for quantifying the error between two dt-SCS and two interconnected dt-SCS without internal signals, respectively. Then we establish an explicit dynamical representation of finite $\widehat\Sigma$ constructed in \cite{SAM15} and show how it can be used to compare interconnections of dt-SCS and those of their finite abstract counterparts based on these new notions. Finally, in the example section, we synthesize policies for abstract dt-SCS locally and refine them back to the original dt-SCS while providing guarantees on the quality of the synthesized policies with respect to satisfaction of local specifications.
This guarantee is compared against the approach of \cite{SAM15} with the metric in \eqref{eq:metric_lit} in the example section.

\section{Stochastic Storage and Simulation Functions}
\label{sec:SPSF}
In this section, we first introduce a notion of so-called stochastic storage functions for dt-SCS with both internal and external inputs, which is adapted from the notion of storage functions from dissipativity theory. We then define a notion of so-called stochastic simulation functions for systems with only external inputs and outputs. We use these definitions to quantify closeness of two dt-SCS.
\begin{definition}\label{Def_1a}
	Consider dt-SCS $\Sigma =(X,U,W,\varsigma,f,Y_1,Y_2, h_1, h_2)$ and
	$\widehat\Sigma =(\hat X,\hat U,\hat W, \varsigma,\hat f, \hat Y_1, \hat Y_2, \hat h_1, \hat h_2)$
	where $\hat Y_1\subseteq Y_1$. A function $V:X\times\hat X\to\R_{\ge0}$ is
	called a stochastic storage function (SStF) from  $\widehat\Sigma$ to $\Sigma$ if there exist
	$\alpha\in\mathcal{K}_\infty$,~$\kappa\in \mathcal{K}$, $\rho_{\mathrm{ext}}\in\mathcal{K}_\infty\cup\{0\}$, constant $\psi \in\R_{\ge 0}$, matrices $G,\hat G,H$ of appropriate dimensions, and symmetric matrix $\bar X$ with conformal block partitions $\bar X^{ij}$, $i,j\in\{1,2\}$, such that
	for any $x\in X$ and $\hat x\in\hat X$ one has
	\begin{align}\label{Eq_2a}
	\alpha(\Vert h_1(x)-\hat h_1(\hat x)\Vert)\le V(x,\hat x),
	\end{align}
	and $\forall\hat\nu\in\hat U$  $\exists \nu\in U$ such that $\forall \hat w\in\hat W$ $\forall w\in W$ one obtains
	\begin{align}\notag\label{Eq_3a}
	\mathbb{E} \Big[V(f(x,\nu,w,\varsigma)&,\hat{f}(\hat x,\hat \nu,\hat w,\varsigma))\,\big|\,x,\hat{x}, \nu,\hat{\nu}, w,\hat w\Big]-V(x,\hat{x})\leq-\kappa(V(x,\hat{x}))\!+\! 
	\rho_{\mathrm{ext}}(\Vert\hat\nu\Vert)\!+\!\psi\\ 
	&\!+\!\begin{bmatrix}
	Gw-\hat G\hat w\\
	h_2(x)-H\hat h_2(\hat x)
	\end{bmatrix}^T
	\underbrace{\begin{bmatrix}
		\bar X^{11}&\bar X^{12}\\
		\bar X^{21}&\bar X^{22}
		\end{bmatrix}}_{\bar X:=}\begin{bmatrix}
	Gw-\hat G\hat w\\
	h_2(x)-H\hat h_2(\hat x)
	\end{bmatrix}\!\!.
	\end{align}
\end{definition}
If there exists an SStF $V$ from $\widehat\Sigma$ to $\Sigma$, this is denoted by $\widehat\Sigma\preceq_{\mathcal{S}}\Sigma$ and the control system $\widehat\Sigma$ is called an abstraction of concrete (original) system $\Sigma$. Note that $\widehat{\Sigma}$ may be finite or infinite depending on cardinalities of sets $\hat X,\hat U,\hat W$.

\begin{remark}
	The last term in inequality~\eqref{Eq_3a} is interpreted in dissipativity theory as the energy supply rate of the system \cite{2016Murat}. Here we choose this function to be quadratic which results in tractable compositional conditions later in the form of linear matrix (in)equalities.
\end{remark} 

\begin{remark}
	The second condition in Definition~\ref{Def_1a} implies implicitly the existence of a function $\nu=\nu_{\hat \nu}(x,\hat x,\hat \nu)$ for the satisfaction of \eqref{Eq_3a}. This function is called the \emph{interface function} and can be used to refine a synthesized policy $\hat\nu$ for $\widehat\Sigma$ to a policy $\nu$ for $\Sigma$. 
\end{remark}

Now, we modify the above notion for the interconnected dt-SCS without internal inputs and outputs.
\begin{definition}
	Consider two dt-SCS
	$\Sigma =(X,U,\varsigma,f, Y,h)$ and
	$\widehat\Sigma =(\hat X,\hat U,\varsigma,\hat f, \hat Y,\hat h)$ without internal inputs and outputs, where $\hat Y\subseteq Y$.
	A function $V:X\times\hat X\to\R_{\ge0}$ is
	called a \emph{stochastic simulation function} (SSF) from $\widehat\Sigma$  to $\Sigma$ if
	\begin{itemize}
		\item there exists $\alpha\in \mathcal{K}_{\infty}$ such that for all  $x\in X$ and $\hat x\in\hat X$,
		\begin{align}\label{eq:lowerbound2}
		\alpha(\Vert h(x)-\hat h(\hat x)\Vert)\le V(x,\hat x),
		\end{align}
		\item for all  $x\in X,\,\hat x\in\hat X,\,\hat\nu\in\hat U$, there exists $\nu\in U$ such that
		\begin{align}\label{eq6666}
		&\mathbb{E} \Big[V(f(x,\nu,\varsigma),\hat{f}(\hat x,\hat \nu,\varsigma))\,\big|\,x,\hat{x},\nu, \hat{\nu}\Big]-V(x,\hat{x})\leq-\kappa(V(x,\hat{x}))
		\!+\!\rho_{\mathrm{ext}}(\Vert\hat\nu\Vert)\!+\!\psi,
		\end{align}
		for some $\kappa\in \mathcal{K}$, $\rho_{\mathrm{ext}} \in \mathcal{K}_{\infty}\cup \{0\}$, and $\psi \in\R_{\ge 0}$.
	\end{itemize}
\end{definition}
If there exists an SSF $V$ from $\widehat\Sigma$ to $\Sigma$, this is denoted by $\widehat\Sigma\preceq\Sigma$ and $\widehat\Sigma$ is called an abstraction of $\Sigma$.


The next theorem shows usefulness of SSF in comparing output trajectories of two dt-SCS in a probabilistic sense.  
\begin{theorem}\label{Thm_1a}
	Let
	$\Sigma =(X,U,\varsigma,f, Y,h)$ and
	$\widehat\Sigma =(\hat X,\hat U,\varsigma,\hat f, \hat Y,\hat h)$
	be two dt-SCS without internal inputs and outputs, where $\hat Y\subseteq Y$.
	Suppose $V$ is an SSF from $\widehat\Sigma$ to $\Sigma$, and there exists a constant $0<\widehat\kappa<1$ such that the function $\kappa \in \mathcal{K}$ in \eqref{eq6666} satisfies $\kappa(r)\geq\widehat\kappa r$ $\forall r\in\R_{\geq0}$. For any external input trajectory $\hat\nu(\cdot)\in\mathcal{\hat U}$ that preserves Markov property for the closed-loop $\widehat\Sigma$, and for any random variables $a$ and $\hat a$ as the initial states of the two dt-SCS,
	there exists an input trajectory $\nu(\cdot)\in\mathcal{U}$ of $\Sigma$ through the interface function associated with $V$ such that the following inequality holds	
	\begin{align}
	&\mathbb{P}\left\{\!\sup_{0\leq k\leq T_d}\!\!\Vert y_{a\nu}(k)-\hat y_{\hat a \hat\nu}(k)\Vert\!\geq\!\varepsilon\,|\,[a;\hat a]\right\}\!\leq\!
	\begin{cases}
	1\!-\!\Big(1\!-\!\frac{V(a,\hat a)}{\alpha\left(\varepsilon\right)}\Big)\Big(1\!-\!\frac{\widehat\psi}{\alpha\left(\varepsilon\right)}\Big)^{T_d} & \text{if}~\alpha\left(\varepsilon\right)\geq\frac{\widehat\psi}{\widehat\kappa},\\
	\Big(\frac{V(a,\hat a)}{\alpha\left(\varepsilon\right)}\Big)(1\!-\!\widehat\kappa)^{T_d}\!+\Big(\frac{\widehat\psi}{\widehat\kappa\alpha\left(\varepsilon\right)}\Big)(1\!-\!(1\!-\!\widehat\kappa)^{T_d}) & \text{if}~\alpha\left(\varepsilon\right)<\frac{\widehat\psi}{\widehat\kappa},
	\end{cases}.\label{Eq_25}
	\end{align}
	where the constant $\widehat\psi\geq0$ satisfies  $\widehat\psi\geq \rho_{\mathrm{ext}}(\Vert\hat \nu\Vert_{\infty})+\psi$.
\end{theorem}

\begin{proof}
	Since $V$ is an SSF from $\widehat\Sigma$ to $\Sigma$, we have
	\begin{align}\notag
	\mathbb{P}&\left\{\sup_{0\leq k\leq T_d}\Vert y_{a\nu}(k)-\hat y_{\hat a \hat\nu}(k)\Vert\geq\varepsilon\,|\,[a;\hat a]\right\}
	=\mathbb{P}\left\{\sup_{0\leq k\leq T_d}\alpha\left(\Vert y_{a\nu}(k)-\hat y_{\hat a \hat\nu}(k)\Vert\right)\geq\alpha(\varepsilon)\,|\,[a;\hat a]\right\}\\\label{eq:supermart}
	\leq&\mathbb{P}\left\{\sup_{0\leq k\leq T_d}V\left(x_{a\nu}(k),\hat x_{\hat a \hat\nu}(k)\right)\geq\alpha(\varepsilon)\,|\,[a;\hat a]\right\}.
	\end{align}
	The equality holds due to $\alpha$ being a $\mathcal K_\infty$ function. The inequality is true due to condition~\eqref{eq:lowerbound2} on the SSF $V$. The results follows by applying Theorem 3 in \cite[pp. 81]{1967stochastic} to \eqref{eq:supermart} and utilizing inequality~\eqref{eq6666}.
\end{proof}
The results shown in Theorem \ref{Thm_1a} provide closeness of output trajectories of two dt-SCS in finite-time horizon. We can extend the result to infinite-time horizon given that $\widehat{\psi}=0$ as presented in the next corollary. 
\begin{corollary}
	Let $\Sigma$ and $\widehat\Sigma$ be two dt-SCS without internal inputs and outputs, where $\hat Y\subseteq Y$.
	Suppose $V$ is an SSF from $\widehat\Sigma$ to $\Sigma$ such that $\rho_{\mathrm{ext}}(\cdot)\equiv0$ and $\psi = 0$.
	For any external input trajectory $\hat\nu(\cdot)\in\mathcal{\hat U}$ preserving Markov property for the closed-loop $\widehat \Sigma$, and for any random variables $a$ and $\hat{a}$ as the initial states of the two dt-SCS, there exists $\nu(\cdot)\in{\mathcal{U}}$ of $\Sigma$ through the interface function associated with $V$ such that the following inequality holds:
	\begin{align}\nonumber
	\mathbb{P}\left\{\sup_{0\leq k< \infty}\Vert y_{a\nu}(k)-\hat y_{\hat a\hat\nu}(k)\Vert\ge \varepsilon\,|\,[a;\hat a]\right\}\leq\frac{V(a,\hat a)}{\alpha\left(\varepsilon\right)}.
	\end{align}
\end{corollary}
\begin{proof}
	Since $V$ is an SSF from $\widehat\Sigma$ to $\Sigma$ with $\rho_{\mathrm{ext}}(\cdot)\equiv0$ and $\psi = 0$, for any $x\in X$ and $\hat x\in\hat X$ and any $\hat\nu\in\hat U$, there exists $\nu\in U$ such that
	\begin{align}\notag
	&\mathbb{E} \Big[V(f(x,\nu,\varsigma),\hat{f}(\hat x,\hat \nu,\varsigma))\,\big|\,x,\hat{x},\nu, \hat{\nu}\Big]-V(x,\hat{x})\leq-\kappa(V(x,\hat{x})),
	\end{align}
	which makes the function $V\left(x_{a\nu}(k),\hat x_{\hat a \hat\nu}(k)\right)$ a nonnegative supermartingale \cite{oksendal2013stochastic} for the joint process $(x_{a\nu}(k),\hat x_{\hat a \hat\nu}(k))$. The rest of the proof follows from the proof of Theorem  \ref{Thm_1a} and the nonnegative supermartingale property \cite{1967stochastic}.
\end{proof}


\section{Compositional Abstractions for Interconnected Systems}
\label{sec:compositionality}
In this section, we analyze networks of stochastic control subsystems and show how to construct their abstractions together with the corresponding simulation functions by using abstractions and stochastic storage functions of the subsystems.

\subsection{Interconnected Stochastic Control Systems}
We first provide a formal definition of interconnection of discrete-time stochastic control subsystems.

\begin{definition}
	Consider $N\in\N_{\geq1}$ stochastic control subsystems $\Sigma_i=(X_i,U_i,W_i,\varsigma_i,f_i,Y_{1i},Y_{2i},h_{1i},h_{2i})$,
	$i\in\{1,\ldots,N\}$,
	and a matrix $M$ defining the coupling between these subsystems. We require the condition $M\prod_{i=1}^N Y_{2i} \subseteq \prod_{i=1}^N W_{i}$ to have a well-posed interconnection. The interconnection of  $\Sigma_i$,
	$\forall i\in \{1,\ldots,N\}$,
	is the dt-SCS $\Sigma=(X,U,\varsigma,f,Y,h)$, denoted by
	$\mathcal{I}(\Sigma_1,\ldots,\Sigma_N)$, such that $X:=\prod_{i=1}^{N}X_i$,  $U:=\prod_{i=1}^{N}U_i$, $f:=\prod_{i=1}^{N}f_{i}$, $Y:=\prod_{i=1}^{N}Y_{1i}$, and $h=\prod_{i=1}^{N}h_{1i}$, with the
	internal inputs constrained according to
	\begin{align}\notag
	[{w_{1};\ldots;w_{N}}]=M[{h_{21}(x_1);\ldots;h_{2N}(x_N)}].
	\end{align}
\end{definition}

In the above definition we allow the interconnection matrix $M$ to have real entries. This is a generalization of composition performed in \cite{2017arXiv170909546M} where the interconnection matrix takes only binary entries.

\subsection{Compositional Abstractions}

We assume that we are given $N$ stochastic control subsystems $\Sigma_i=(X_i,U_i,W_i,\varsigma_i,f_i,Y_{1i},Y_{2i},h_{1i},h_{2i})$ together with their corresponding abstractions $\widehat\Sigma_i\!=\!(\hat X_i,\hat U_i,\hat W_i, \varsigma_i,\hat f_i,\hat Y_{1i},\hat Y_{2i},\\,\hat h_{1i}, \hat h_{2i})$ with SStF $V_i$ from $\widehat\Sigma_i$ to $\Sigma_i$. Indicate by $\alpha_{i}$, $\kappa_i$, $\rho_{i\mathrm{ext}}$, $H_i$, $G_i$, $\hat G_i$, $\bar X_i$, $\bar X_i^{11}$, $\bar X_i^{12}$, $\bar X_i^{21}$, and $\bar X_i^{22}$, the corresponding functions, matrices, and the conformal block partitions appearing in Definition \ref{Def_1a}.
In the next theorem, as one of the main results of the paper, we provide sufficient conditions for having an SSF from the interconnection of abstractions $\widehat \Sigma=\mathcal{I}(\widehat\Sigma_1,\ldots,\widehat\Sigma_N)$ to that of concrete ones
$\Sigma=\mathcal{I}(\Sigma_1,\ldots,\Sigma_N)$. This theorem enables us to quantify in probability the error between the interconnection of stochastic control subsystems and that of their abstractions in a compositional manner by leveraging Theorem~\ref{Thm_1a}.
\begin{theorem}\label{Thm_2a}
	Consider the interconnected stochastic control system
	$\Sigma=\mathcal{I}(\Sigma_1,\ldots,\Sigma_N)$ induced by $N\in{\N_{\geq1}}$ stochastic
	control subsystems~$\Sigma_i$ and the coupling matrix $M$. Suppose that each stochastic control subsystem $\Sigma_i$ admits an abstraction $\widehat \Sigma_i$ with the corresponding SStF $V_i$. 	%
	Then the weighted sum
	\begin{equation}
	\label{eq:V_comp}
	V(x,\hat x)\Let\sum_{i=1}^N\mu_iV_i(x_i,\hat x_i)
	\end{equation}
	is a stochastic simulation function from the interconnected control system
	$\widehat \Sigma=\mathcal{I}(\widehat \Sigma_1,\ldots,\widehat\Sigma_N)$, with coupling matrix $\hat M$, to $\Sigma=\mathcal{I}(\Sigma_1,\ldots,\Sigma_N)$
	if $\mu_{i}>0$, $i\in\{1,\ldots,N\}$, and $\hat M$ satisfy matrix (in)equality and inclusion
	\begin{align}
	\begin{bmatrix}\label{Con_1a}
	GM\\I_{\tilde q}
	\end{bmatrix}^T \bar X&_{cmp}\begin{bmatrix}
	GM\\I_{\tilde q}
	\end{bmatrix}\preceq0,
	\\\label{Con_2a}
	GMH&=\hat G\hat M,\\
	\hat M\prod_{i=1}^N \hat Y_{2i} &\subseteq \prod_{i=1}^N \hat W_{i},\label{Con111}
	\end{align}
	where
	\begin{align}\notag
	\quad G:=\mathsf{diag}&(G_1,\ldots,G_N),~\hat G:=\mathsf{diag}(\hat G_1,\ldots,\hat G_N),\qquad\!\!\!\!\!\!\!\!\!\! H:=\mathsf{diag}(H_1,\ldots,H_N),\\\label{Def_3a}
	&\bar X_{cmp}\!\!:=\begin{bmatrix}
	\mu_1\bar X_1^{11}&&&\mu_1\bar X_1^{12}&&\\
	&\ddots&&&\ddots&\\
	&&\mu_N\bar X_N^{11}&&&\mu_N\bar X_N^{12}\\
	\mu_1\bar X_1^{21}&&&\mu_1\bar X_1^{22}&&\\
	&\ddots&&&\ddots&\\
	&&\mu_N\bar X_N^{21}&&&\mu_N\bar X_N^{22}
	\end{bmatrix}\!\!,
	\end{align}
	and $\tilde q=\sum_{i=1}^Nq_{2i}$ with $q_{2i}$ being the internal output dimensions of subsystems $\Sigma_i$.
\end{theorem}
\begin{figure*}[ht]
	\rule{\textwidth}{0.1pt}
	\begin{align}\notag
	&\mathbb{E}\Big[\!\!\sum_{i=1}^N\!\mu_i\Big[V_i(f_i(x_i,\nu_i,w_i,\varsigma_i),\hat{f}_i(\hat x_i,\hat \nu_i,\hat w_i,\varsigma_i))|x,\hat x,\hat{\nu}\Big]\Big]\!\!-\!\!\sum_{i=1}^N\!\mu_iV_i(x_i,\hat x_i)\!=\!\!\sum_{i=1}^N\!\mu_i\mathbb{E}\Big[V_i(f_i(x_i,\nu_i,w_i,\varsigma_i),\hat{f}_i(\hat x_i,\hat \nu_i,\hat w_i,
	\\\notag
	&, \varsigma_i))\,|\,x,\hat x,\hat{\nu}\Big]\!-\! \sum_{i=1}^N\mu_iV_i(x_i,\hat x_i)\!=\! \sum_{i=1}^N\mu_i\mathbb{E}\Big[V_i(f_i(x_i,\nu_i,w_i,\varsigma_i),\hat{f}_i(\hat x_i,\hat \nu_i,\hat w_i,\varsigma_i))\,|\,x_i,\hat x_i,\hat{\nu}_i\Big]\!-\!\sum_{i=1}^N\mu_iV_i(x_i,\hat x_i)
	\\\notag
	&\leq\sum_{i=1}^N\mu_i\bigg(\!\!-\!\kappa_i(V_i( x_i,\hat x_i))\!+\!\rho_{i\mathrm{ext}}(\Vert \hat \nu_i\Vert)+\psi_i\!+\!\begin{bmatrix}
	G_iw_i-\hat G_i\hat w_i\\
	h_{2i}(x_i)-H_i\hat h_{2i}(\hat x_i)
	\end{bmatrix}^T\begin{bmatrix}
	\bar X_i^{11}&\bar X_i^{12}\\
	\bar X_i^{21}&\bar X_i^{22}
	\end{bmatrix}\begin{bmatrix}
	G_iw_i-\hat G_i\hat w_i\\
	h_{2i}(x_i)-H_i\hat h_{2i}(\hat x_i)
	\end{bmatrix}\bigg)
	\\\notag
	&=\sum_{i=1}^N-\mu_i\kappa_i(V_i( x_i,\hat x_i))+\sum_{i=1}^N\mu_i\rho_{i\mathrm{ext}}(\Vert \hat \nu_i\Vert)+\sum_{i=1}^N\mu_i\psi_i+\begin{bmatrix}
	G_1w_1-\hat G_1\hat w_1\\
	\vdots\\
	G_Nw_N-\hat G_N\hat w_N\\
	h_{21}(x_1)-H_1\hat h_{21}(\hat x_1)\\
	\vdots\\
	h_{2N}(x_N)-H_N\hat h_{2N}(\hat x_N)
	\end{bmatrix}^T\\\notag
	&\begin{bmatrix}
	\mu_1\bar X_1^{11}&&&\mu_1\bar X_1^{12}&&\\
	&\ddots&&&\ddots&\\
	&&\mu_N\bar X_N^{11}&&&\mu_N\bar X_N^{12}\\
	\mu_1\bar X_1^{21}&&&\mu_1\bar X_1^{22}&&\\
	&\ddots&&&\ddots&\\
	&&\mu_N\bar X_N^{21}&&&\mu_N\bar X_N^{22}
	\end{bmatrix}\begin{bmatrix}
	G_1w_1-\hat G_1\hat w_1\\
	\vdots\\
	G_Nw_N-\hat G_N\hat w_N\\
	h_{21}(x_1)-H_1\hat h_{21}(\hat x_1)\\
	\vdots\\
	h_{2N}(x_N)-H_N\hat h_{2N}(\hat x_N)
	\end{bmatrix}=
	\!\sum_{i=1}^N-\mu_i\kappa_i(V_i( x_i,\hat x_i))\\\notag
	&\!+\!\!\sum_{i=1}^N\mu_i\rho_{i\mathrm{ext}}(\Vert \hat \nu_i\Vert)\!+\!\!\sum_{i=1}^N\mu_i\psi_i
	\!+\!\!\begin{bmatrix}
	\!GM\!\begin{bmatrix}
	h_{21}(x_1)\\
	\vdots\\
	h_{2N}(x_N)
	\end{bmatrix}\!\!-\hat G \hat M\!\begin{bmatrix}
	\hat h_{21}(\hat x_1)\\
	\vdots\\
	\hat h_{2N}(\hat x_N)
	\end{bmatrix}\\
	h_{21}(x_1)-H_1\hat h_{21}(\hat x_1)\\
	\vdots\\
	h_{2N}(x_N)-H_N\hat h_{2N}(\hat x_N)
	\end{bmatrix}^T\!\!\!\!\!\!\bar X_{cmp}\!\!\begin{bmatrix}
	\!GM\!\begin{bmatrix}
	h_{21}(x_1)\\
	\vdots\\
	h_{2N}(x_N)
	\end{bmatrix}\!-\hat G \hat M\!\begin{bmatrix}
	\hat h_{21}(\hat x_1)\\
	\vdots\\
	\hat h_{2N}(\hat x_N)
	\end{bmatrix}\\
	h_{21}(x_1)-H_1\hat h_{21}(\hat x_1)\\
	\vdots\\
	h_{2N}(x_N)-H_N\hat h_{2N}(\hat x_N)
	\end{bmatrix}\\\notag
	&=\sum_{i=1}^N-\mu_i\kappa_i(V_i( x_i,\hat x_i))+\sum_{i=1}^N\mu_i\rho_{i\mathrm{ext}}(\Vert \hat \nu_i\Vert)+\sum_{i=1}^N\mu_i\psi_i
	+\begin{bmatrix}
	h_{21}(x_1)-H_1\hat h_{21}(\hat x_1)\\
	\vdots\\
	h_{2N}(x_N)-H_N\hat h_{2N}(\hat x_N)
	\end{bmatrix}^T\begin{bmatrix}
	GM\\
	I_{\tilde q}
	\end{bmatrix}^T\!\!\!\bar X_{cmp}\begin{bmatrix}
	GM\\
	I_{\tilde q}
	\end{bmatrix}\\\label{Eq_4a}
	&\begin{bmatrix}
	\!\!h_{21}(x_1)\!-\!H_1\hat h_{21}(\hat x_1)\!\!\\
	\vdots\\
	\!h_{2N}(x_N)\!-\!H_N\hat h_{2N}(\hat x_N)
	\end{bmatrix}\!\!\leq\!\!\sum_{i=1}^N\!\!-\mu_i\kappa_i(\!V_i( x_i,\hat x_i))\!+\!\!\sum_{i=1}^N\!\mu_i\rho_{i\mathrm{ext}}(\Vert \hat \nu_i\Vert)\!+\!\!\sum_{i=1}^N\!\mu_i\psi_i\!\leq\!\!\!-\kappa\left(\!V\!\left( x,\hat{x}\right)\right)\!\!+\!\!\rho_{\mathrm{ext}}(\left\Vert \hat \nu\right\Vert)\!\!+\!\!\psi.
	\end{align}
	\rule{\textwidth}{0.1pt}
	\vspace{-0.6cm}
\end{figure*}
\begin{proof}
	We first show that SSF $V$ in \eqref{eq:V_comp} satisfies the inequality \eqref{eq:lowerbound2} for some $\mathcal{K}_\infty$ function $\alpha$. For any $x=[{x_1;\ldots;x_N}]\in X$ and  $\hat x=[{\hat x_1;\ldots;\hat x_N}]\in \hat X$, one gets:
	\begin{align}\notag
	\Vert h(x)-\hat h(\hat x) \Vert\!=\!\Vert [h_{11}(x_1);\ldots;&h_{1N}(x_N)]\!-\![\hat h_{11}(\hat x_1);\ldots;\hat h_{1N}(\hat x_N)]\Vert\le\sum_{i=1}^N \Vert  h_{1i}(\hat x_i)-\hat h_{1i}(x_i) \Vert\\\notag
	&\le \sum_{i=1}^N \alpha_{i}^{-1}(V_i( x_i, \hat x_i))\le \bar\alpha(V(x,\hat x)),
	\end{align}
	with function $\bar\alpha:\mathbb R_{\ge 0}\rightarrow\mathbb R_{\ge 0}$ defined for all $r\in\mathbb R_{\ge 0}$ as
	
	\begin{center}
		$\bar\alpha(r) \Let \max\left\{\sum_{i=1}^N\alpha_{i}^{-1}(s_i)\,\,\big|\, s_i  {\ge 0},\,\,\sum_{i=1}^N \mu_i s_i=r\right\}. $
	\end{center}
	
	
It is not hard to verify that function $\bar\alpha(\cdot)$ defined above is a $\mathcal{K}_\infty$ function.
By taking the $\mathcal{K}_\infty$ function $\alpha(r):=\bar\alpha^{-1}(r)$, $\forall r\in\R_{\ge0}$, one obtains
$$\alpha(\Vert h(x)-\hat h(\hat x)\Vert)\le V( x, \hat x),$$
satisfying inequality \eqref{eq:lowerbound2}.
Now we prove that SSF $V$ in \eqref{eq:V_comp} satisfies inequality \eqref{eq6666}.
Consider any $x=[{x_1;\ldots;x_N}]\in X$, $\hat x=[{\hat x_1;\ldots;\hat x_N}]\in \hat X$, and
$\hat \nu=[{\hat \nu_{1};\ldots;\hat \nu_{N}}]\in\hat U$. For any $i\in\{1,\ldots,N\}$, there exists $\nu_i\in U_i$, consequently, a vector $\nu=[{\nu_{1};\ldots;\nu_{N}}]\in U$, satisfying~\eqref{Eq_3a} for each pair of subsystems $\Sigma_i$ and $\widehat\Sigma_i$
with the internal inputs given by $[{w_1;\ldots;w_N}]=M[h_{21}(x_1);\ldots;h_{2N}(x_N)]$ and $[{\hat w_1;\ldots;\hat w_N}]=\hat M[\hat h_{21}(\hat x_1);\ldots;\hat h_{2N}(\hat x_N)]$.
Then we have the chain of inequalities in \eqref{Eq_4a}
using conditions \eqref{Con_1a} and \eqref{Con_2a} and by defining $\kappa(\cdot),\rho_{\mathrm{ext}}(\cdot),\psi$ as
\begin{align}\notag
\kappa(r) &\Let \min\left\{\sum_{i=1}^N\mu_i\kappa_i(s_i)\,\,\big|\, s_i  {\ge 0},\,\,\sum_{i=1}^N \mu_i s_i=r\right\} \\\notag
\rho_{\mathrm{ext}}(r) \!&\Let \max\left\{\sum_{i=1}^N\mu_i\rho_{i\mathrm{ext}}(s_i)\,\,\big|\, s_i  {\ge 0},\,\,\|[{s_1;\ldots;s_N}\| = r\right\}]\\\notag
\psi&\Let\sum_{i=1}^N\mu_i\psi_i.
\end{align}
Note that $\kappa$ and $\rho_{\mathrm{ext}}$ in \eqref{Eq_4a} belong to $\mathcal{K}$ and $\mathcal{K}_\infty\cup\{0\}$, respectively, because of their definition provided above. Hence, we conclude that $V$ is an SSF from $\widehat \Sigma$ to $\Sigma$. 
\end{proof}

Figure~\ref{Fig1} illustrates schematically the result of Theorem~\ref{Thm_2a}.

\begin{remark}
	Note that condition \eqref{Con_1a} with $G=I$ is exactly similar to the linear matrix inequality (LMI) appeared in \cite{2016Murat} as composotional stability condition based on dissipativity theory. As discussed in \cite{2016Murat}, the LMI holds independently of the number of subsystems in many physical applications with specific interconnection structures including communication networks, flexible joint robots, and power generators.
\end{remark}

\begin{remark}
	For the compositional construction of finite gMDPs provided in the next section, condition \eqref{Con_2a} is satisfiable by simply selecting  $\hat M = M$.
	Moreover, condition \eqref{Con111} is not restrictive for the results provided in the next section since $\hat W_i$ and $\hat Y_{2i}$ are internal input and output sets of the abstract subsystems $\widehat\Sigma_i$, which are finite. Thus one can readily choose internal input sets $\hat W_i$ such that $\prod_{i=1}^n \hat W_{i} := \hat M\prod_{i=1}^n \hat Y_{2i}$ which implicitly implies a condition on the granularity of discretization for sets $W_i$ and $Y_{2i}$.
\end{remark}

\begin{figure}[ht]
	\begin{center}
		\includegraphics[scale = 0.25]{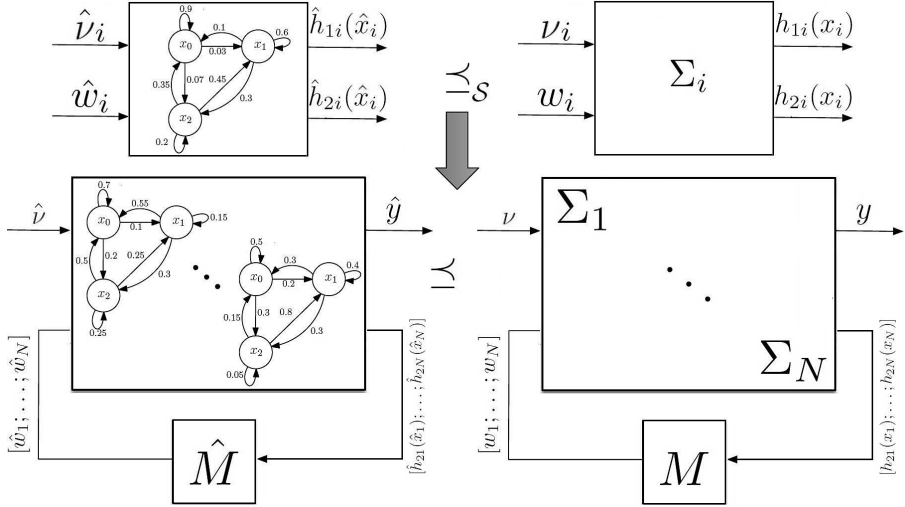}
		\caption{Compositionality results provided that conditions \eqref{Con_1a}, \eqref{Con_2a}, and \eqref{Con111} are satisfied.}
		\label{Fig1}
	\end{center}
\end{figure}

\section{Construction of Finite Markov Decision Processes}
\label{sec:constrcution_finite}

In the previous sections, $\Sigma$ and $\widehat{\Sigma}$ were considered as general discrete-time stochastic control systems without discussing the cardinality of their state spaces. In this section, we consider $\Sigma$ as an infinite dt-SCS and $\widehat{\Sigma}$ as its finite abstraction constructed as in Section~\ref{subsec:MDP}.
We impose conditions on the infinite dt-SCS $\Sigma$ enabling us to find SStF from its finite abstraction $\widehat{\Sigma}$ to $\Sigma$. The required conditions are first presented in a general setting for nonlinear stochastic control systems in Section~\ref{subsec:nonlinear} and then represented via some matrix inequality for linear stochastic control systems in Section~\ref{subsec:linear}.


\subsection{Discrete-Time Nonlinear Stochastic Control Systems}
\label{subsec:nonlinear}

The stochastic storage function from finite MDP $\widehat{\Sigma}$ of Section \ref{subsec:MDP} to $\Sigma$ is established under the assumption that the original discrete-time stochastic control system $\Sigma$ is so-called incrementally passivable as in Assumption~\ref{Def111}.

\begin{assumption}\label{Def111}
	A dt-SCS 
	$\Sigma=(X, U, W,\varsigma, f, Y_1, Y_{2},h_1,h_2)$
	is called \emph{incrementally passivable} if there exist functions $ L: X \to U$ and $ V: X \times X \to \mathbb{R}_{\geq0} $  such that $\forall x,x'\in X$, $\forall \nu\in U$, $\forall w,w' \in W$, the inequalities:
	
	\begin{align}\label{Con555}
	\underline{\alpha} (\Vert h_1(x)-h_1(x')\Vert ) \leq V(x,x'),
	\end{align}	
	and	
	\begin{align}\notag
	\mathbb{E} \Big[V(f(x,L(x)+\nu,&w,\varsigma),f(x',L(x')+\nu,w',\varsigma))\big|x, x',\nu, w, w'\Big]\!-\!V(x,x')\leq\!-\widehat{\kappa}(V(x,x'))\\\label{Con854}
	&+\begin{bmatrix}w-w'\\
	\!h_2(x)-h_2(x')\!
	\end{bmatrix}^T\!\overbrace{\begin{bmatrix}
		\bar X^{11}&\bar X^{12}\\
		\bar X^{21}&\bar X^{22}
		\end{bmatrix}}^{\bar X:=}\begin{bmatrix}
	w-w'\\
	\!h_2(x)-h_2(x')\!
	\end{bmatrix}\!\!,
	\end{align}
	hold for some $\underline{\alpha}\in \mathcal{K}_{\infty}$, $\widehat{\kappa}\in \mathcal{K}$, and matrix $\bar X$ of appropriate dimension.
\end{assumption}

	\begin{remark}
	Note that Assumption \ref{Def111} implies that $V$ is a SStF from system $\Sigma$ equipped with the state feedback controller $L$ to itself. This type of property is closely related to the notion of so-called incremental stabilizability \cite{angeli,pham2009contraction}.
\end{remark}
In Section~\ref{subsec:linear}, we show that inequalities \eqref{Con555}-\eqref{Con854} for a candidate quadratic function $V$ and linear stochastic control systems boil down to some matrix inequality.

%
%


Under Assumption~\ref{Def111}, the next theorem shows a relation between $\Sigma$ and $\widehat{\Sigma}$, constructed as in Algorithm \ref{algo:MC_app}, via establishing a stochastic storage function between them.
\begin{theorem}
	Let $\Sigma$ be an incrementally passivable dt-SCS  via a function $V$ as in Assumption~\ref{Def111} and $\widehat{\Sigma}$ be a finite MDP as in Algorithm \ref{algo:MC_app}. Assume that there exists a function $\gamma\in\mathcal{K}_{\infty}$  such that $V$ satisfies
	\begin{equation}\label{Eq65}
	V(x,x')-V(x,x'')\leq \gamma(\Vert x'-x''\Vert),\quad \forall x,x',x'' \in X.
	\end{equation}
	Then $V$ is a stochastic storage function from $\widehat{\Sigma}$ to $\Sigma$.
\end{theorem}

\begin{proof} Since system $\Sigma$ is incrementally passivable, from \eqref{Con555}, $\forall x\in X$ and $ \forall \hat x \in \hat X
	$, we have 
	\begin{align}\notag
	\underline{\alpha} (\Vert h_1(x)-h_1(\hat x ) \Vert )= \underline{\alpha} (\Vert h_1(x)-\hat{h}_1(\hat{x})\Vert ) \leq V(x,\hat{x}),
	\end{align}
	satisfying \eqref{Eq_2a} with  $\alpha(s) \Let \underline{\alpha}(s) $ $\forall s\in \R_{\geq0}$.
	Now by taking the conditional expectation from \eqref{Eq65}, $\forall x\in X, \forall \hat{x} \in \hat X, \forall \hat{\nu} \in \hat U,\forall w \in W,\forall \hat{w} \in \hat W$, we have 
	\begin{align}\notag
	&\mathbb{E}\Big[V(f(x,L(x)\!+\!\hat{\nu},w,\varsigma),\hat f(\hat{x},\hat{\nu},\hat{w},\varsigma))\big|x,\hat x,\hat \nu, w, \hat w\Big]\!-\!\mathbb{E}\Big[V(f(x,L(x)\!+\!\hat{\nu},w,\varsigma),f(\hat{x},L(\hat{x})+\hat{\nu},\hat{w},\varsigma))\big|x,\hat x,\hat \nu, w, \hat w\Big]\\\notag&\leq\mathbb{E}\Big[\gamma(\Vert\hat f(\hat{x},\hat{\nu},\hat{w},\varsigma)-f(\hat{x},L(\hat{x})+\hat{\nu},\hat{w},\varsigma)\Vert)\big|\hat x,\hat x,\hat \nu, w, \hat w\Big],
	\end{align}
	where $\hat f(\hat{x},\hat{\nu},\hat{w},\varsigma) = \Pi(f(\hat{x},L(\hat{x})+\hat{\nu},\hat{w},\varsigma))$. Using Theorem~\ref{Def154} and inequality~\eqref{eq:Pi_delta}, the above inequality reduces to
	\begin{align}\notag
	&\mathbb{E}\Big[V(f(x,L(x)\!+\!\hat{\nu},w,\varsigma),\hat f(\hat{x},\hat{\nu},\hat{w},\varsigma))\big|x,\hat x,\hat \nu, w, \hat w\Big]\!-\!\mathbb{E}\Big[V(f(x,L(x)\!+\!\hat{\nu},w,\varsigma),f(\hat{x},L(\hat{x})+\hat{\nu},\hat{w},\varsigma))\big|x,\hat x,\hat \nu, w, \hat w\Big]\\\notag&\leq\gamma(\delta).
	\end{align}
	
	Employing \eqref{Con854} and since $h_2=\hat h_2$, we get 
	\begin{align}\notag
	\mathbb{E}\Big[V(f(x,L(x)+&\hat{\nu},w,\varsigma),f(\hat{x},L(\hat{x})+\hat{\nu},\hat{w},\varsigma))\big|x,\hat x,\hat \nu,w,\hat w\Big]-V(x,\hat{x})\leq-\widehat{\kappa}(V(x,\hat{x}))\\\notag
	&+\begin{bmatrix}w-\hat w\\\notag
	h_2(x)-\hat h_2(\hat x)
	\end{bmatrix}^T\begin{bmatrix}
	\bar X^{11}&\bar X^{12}\\
	\bar X^{21}&\bar X^{22}
	\end{bmatrix}\begin{bmatrix}
	w-\hat w\\
	h_2(x)-\hat h_2(\hat x)
	\end{bmatrix}\!\!.
	\end{align}
	It follows that $\forall x \in X, \forall \hat x \in \hat X, \forall \hat u \in U,$ and $\forall w \in W,\forall \hat w \in \hat W $, one obtains
	
	\begin{align}\notag
	\mathbb{E}\Big[V(f(x,L(x)+&\hat{\nu},w,\varsigma),\hat f(\hat{x},\hat{\nu},\hat{w},\varsigma)))\big|x,\hat x,\hat \nu, w, \hat w\Big]-V(x,\hat{x})\leq-\widehat{\kappa}(V(x,\hat{x}))+\gamma(\delta)\\\notag
	&+\begin{bmatrix}w-\hat w\\\notag
	h_2(x)-\hat h_2(\hat x)
	\end{bmatrix}^T\begin{bmatrix}
	\bar X^{11}&\bar X^{12}\\
	\bar X^{21}&\bar X^{22}
	\end{bmatrix}\begin{bmatrix}
	w-\hat w\\
	h_2(x)-\hat h_2(\hat x)
	\end{bmatrix}\!\!,
	\end{align}
	satisfying \eqref{Eq_3a}
	with $\psi=\gamma(\delta)$, $\nu=L(x)+\hat{\nu}$, $\kappa=\widehat{\kappa}$, $\rho_{ext}\equiv 0,$ and $G,$ $ \hat{G},$ $H $ are identity matrices of appropriate dimensions. Hence, $V$ is an SStF from $\widehat \Sigma$ to $\Sigma$, which completes the proof.  		
\end{proof}

\begin{remark}
	As shown in \cite{zamani2014symbolic} and by employing the mean value theorem, assumption \eqref{Eq65} is not restrictive provided that $V$ is restricted to a compact subset of $X \times X$.
\end{remark}
Now we provide similar results as in Subsection \ref{subsec:nonlinear} but tailored to linear stochastic control systems.
\subsection{Discrete-Time Linear Stochastic Control Systems}
\label{subsec:linear}
In this subsection, we focus on the class of discrete-time linear stochastic control systems $\Sigma$ and \emph{quadratic} stochastic storage functions $V$. First, we formally define the class of discrete-time linear stochastic control systems. Afterwards, we construct their finite Markov decision processes $\widehat \Sigma$ as in Theorem \ref{Def154}, and then provide conditions under which a candidate V is an SStF from $\widehat \Sigma$ to $\Sigma$. 

The class of discrete-time linear stochastic control systems is given by
\begin{align}\label{Eq_5a}
\Sigma\!:\!\left\{\begin{array}{l}x(k+1)\!=\!Ax(k)\!+\!B\nu(k)\!+\!Dw(k)\!+\!N\varsigma(k),\\
y_1(k)=C_1x(k),\\
y_2(k)=C_2x(k),\end{array}\right.
\end{align}
where the additive noise $\varsigma(k)$ is a sequence of independent random vectors with multivariate standard normal distributions.

We use the tuple
\begin{align}\notag
\Sigma=(A,B,C_1,C_2,D,N),
\end{align}
to refer to the class of discrete-time linear stochastic control systems of the form~\eqref{Eq_5a}.

Consider the following quadratic function
\begin{align}\label{Eq_7a}
V(x,\hat x)=(x-\hat x)^T\tilde M(x-\hat x),
\end{align}
where $\tilde M$ is a positive-definite matrix of appropriate dimension. In order to show that $V$ in \eqref{Eq_7a} is an SStF from $\widehat\Sigma$ to $\Sigma$, we require the following key assumption on $\Sigma$. 
\begin{assumption}\label{As_1a}
	Let $\Sigma=(A,B,C_1,C_2,D,N)$. Assume that for some constant
	$0<\widehat\kappa<1$ and $\pi>0$ there exist matrices $\tilde M\succ0$, $K$, $\bar X^{11}$, $\bar X^{12}$, $\bar X^{21}$, and $\bar X^{22}$ of appropriate dimensions such that matrix inequality \eqref{Eq_8a} holds.
	
	\begin{figure*}
		\begin{align}
		\begin{bmatrix}\label{Eq_8a}
		(1+\pi)(A+BK)^T\tilde M(A+BK) && (A+BK)^T\tilde MD \\
		D^T\tilde M (A+BK)&& (1+\pi)D^T \tilde MD \\
		\end{bmatrix}\preceq\begin{bmatrix}
		\widehat\kappa\tilde M+C_2^T\bar X^{22}C_2 & C_2^T\bar X^{21}\\
		\bar X^{12}C_2 & \bar X^{11}\\
		\end{bmatrix}
		\end{align}
		\rule{\textwidth}{0.1pt}
		\begin{align}
		&\notag\mathbb{E} \Big[V(f(x,\nu,w,\varsigma),\hat{f}(\hat x,\hat \nu,\hat w,\varsigma))\,\big|\,x,\hat{x},\hat{\nu}, w,\hat w\Big]\!-\!V(x,\hat x)=(x-\hat x)^T(A+BK)^T\tilde M(A+BK)(x-\hat x)+2 (x-\hat x)^T\\\notag
		&(A+BK)^T\tilde MD(w-\hat w)\!+\!(w-\hat w)^T D^T\tilde M D(w-\hat w)\!+\!2(x-\hat x)^T(A+BK)^T\tilde M\mathbb{E} \Big[F\,|\,x,\hat x , \hat \nu, w,\hat w \Big]\!+\!2 (w-\hat w)^T \\\notag
		&D^T\tilde M\mathbb{E} \Big[F\,|\,x,\hat x, \hat \nu, w,\hat w\Big]\!+\!\mathbb{E} \Big[F^T \tilde M F \,|\,x,\hat x, \hat \nu, w,\hat w\Big]-V(x,\hat x)\\\notag
		&\le
		\begin{bmatrix}x-\hat x\\w-\hat w\\\end{bmatrix}^T\begin{bmatrix}
		(1+\pi)(A+BK)^T\tilde M(A+BK) && (A+BK)^T\tilde MD\\
		D^T\tilde M (A+BK)&& (1+\pi)D^T \tilde MD\\
		\end{bmatrix}\begin{bmatrix}x-\hat x\\w-\hat w\\\end{bmatrix}+ (1+2/\pi)\lambda_{\max}{(\tilde M)}\delta^2-V(x,\hat x)\\\notag
		&\le\begin{bmatrix}x-\hat x\\w-\hat w\\\end{bmatrix}^T\begin{bmatrix}
		\widehat\kappa\tilde M+C_2^T\bar X^{22}C_2 & C_2^T\bar X^{21}\\
		\bar X^{12}C_2 & \bar X^{11}\\
		\end{bmatrix}\begin{bmatrix}x-\hat x\\w-\hat w\\\end{bmatrix}+ (1+2/\pi)\lambda_{\max}{(\tilde M)}\delta^2-V(x,\hat x)\\\label{Eq_55a}
		&=
		-(1-\widehat\kappa) (V(x,\hat x))+\begin{bmatrix}w-\hat w\\C_2x- \hat C_2\hat x\end{bmatrix}^T\begin{bmatrix}
		\bar X^{11}&\bar X^{12}\\
		\bar X^{21}&\bar X^{22}
		\end{bmatrix}\begin{bmatrix}w-\hat w\\C_2x- \hat C_2\hat x\end{bmatrix}+ (1+2/\pi)\lambda_{\max}{(\tilde M)}\delta^2.
		\end{align}
		\rule{\textwidth}{0.1pt}
		\vspace{-5mm}
	\end{figure*}
\end{assumption}
Now, we provide another main result of this section showing that under some conditions $V$ in \eqref{Eq_7a} is an SStF from $\widehat \Sigma$ to $\Sigma$.
\begin{theorem}\label{Thm_3a}
	Let $\Sigma=(A,B,C_1,C_2,D,N)$ and $\widehat \Sigma$ be a  finite Markov decision process with discretization parameter $\delta$, and $\hat Y_1\subseteq Y_1$. Suppose Assumption \ref{As_1a} holds, $C_1 = \hat C_1$, and  $C_2 = \hat C_2$, then function $V$ defined in~\eqref{Eq_7a} is an SStF from $\widehat \Sigma$ to $\Sigma$.
\end{theorem}
\begin{proof}
	Here, we show that $\forall x$, $\forall \hat x$, $\forall \hat \nu$, $\exists\nu$, $\forall \hat w$, $\forall w$, $V$ satisfies $\frac{\lambda_{\min}(\tilde M)}{\lambda_{\max}(C_1^TC_1)}\Vert C_1x-\hat C_1\hat x\Vert^2\le V(x,\hat x)$ and
	\begin{align}\notag
	\mathbb{E}\Big[V(f(x,\nu,w,\varsigma),\hat{f}(\hat x&,\hat \nu,\hat w,\varsigma))\,\big|\,x,\hat{x},\hat{\nu}, w,\hat w\Big]-V(x,\hat x)\leq-(1-\widehat\kappa) (V(x,\hat x))+(1+2/\pi)\lambda_{\max}{(\tilde M)}\delta^2\\\notag
	&+\begin{bmatrix}
	w-\hat w\\
	h_2(x)-\hat h_2(\hat x)
	\end{bmatrix}^T{\begin{bmatrix}
		\bar X^{11}&\bar X^{12}\\
		\bar X^{21}&\bar X^{22}
		\end{bmatrix}}\begin{bmatrix}
	w-\hat w\\
	h_2(x)-\hat h_2(\hat x)
	\end{bmatrix}\!.
	\end{align}
	
	Since $C_1 = \hat C_1$, we have $\Vert C_1x-\hat C_1\hat x\Vert^2=(x-\hat x)^TC_1^TC_1(x-\hat x)$. Since  $\lambda_{\min}(C_1^TC_1)\Vert x- \hat x\Vert^2\leq(x-\hat x)^TC_1^TC_1(x-\hat x)\leq\lambda_{\max}(C_1^TC_1)\Vert x- \hat x\Vert^2$ and similarly  $\lambda_{\min}(\tilde M)\Vert x- \hat x\Vert^2\leq(x-\hat x)^T\tilde M(x-\hat x)\leq\lambda_{\max}(\tilde M)\Vert x- \hat x\Vert^2$, it can be readily verified that  $\frac{\lambda_{\min}(\tilde M)}{\lambda_{\max}(C_1^TC_1)}\Vert C_1x-\hat C_1\hat x\Vert^2\le V(x,\hat x)$ holds $\forall x$, $\forall \hat x$, implying that inequality \eqref{Eq_2a} holds with $\alpha(s)=\frac{\lambda_{\min}(\tilde M)}{\lambda_{\max}(C_1^TC_1)}s^2$ for any $s\in\R_{\geq0}$. We proceed with showing that the inequality~\eqref{Eq_3a} holds, as well. Given any $x$, $\hat x$, and $\hat \nu$, we choose $\nu$ via the following \emph{interface} function:
	\begin{align}\label{Eq_255}
	\nu=\nu_{\hat \nu}(x,\hat x,\hat \nu):=K(x-\hat x)+\hat \nu.
	\end{align}
	By employing the definition of the interface function, we simplify
	\begin{align}\notag
	Ax \!+\! B\nu_{\hat \nu}(x,\hat x, \hat \nu) \!+\! Dw \!+\! N\varsigma \!-\!\Pi(A\hat x \!+\! B\hat \nu \!+\! D\hat w \!+\! N\varsigma)
	\end{align}
	to 
	\begin{align}\notag
	(A+BK)(x-\hat x)\!+\! D(w-\hat w) \!+ F,
	\end{align}
	where $F = \! A\hat x \!+B\hat \nu\!+\! D\hat w \!+ \! N\varsigma-\Pi(A\hat x+ B\hat \nu + D\hat w + N\varsigma)$.
	Using Young's inequality \cite{young1912classes} as $ab\leq \frac{\pi}{2}a^2+\frac{1}{2\pi}b^2,$ for any $a,b\geq0$ and any $\pi>0$, and by employing Cauchy-Schwarz inequality, $C_2 = \hat C_2$, and since 
	\begin{align}\notag
	\left\{\begin{array}{l}\Vert F\Vert\leq \delta,\\
	F^T \tilde M F \leq\lambda_{\max}(\tilde M)\delta^2,\end{array}\right.
	\end{align}
	one can obtain the chain of inequalities in \eqref{Eq_55a}. Hence, the proposed $V$ in \eqref{Eq_7a} is an SStF from  $\widehat \Sigma$ to $\Sigma$, which completes the proof. Note that functions $\alpha\in\mathcal{K}_\infty$, $\kappa\in\mathcal{K}$, $\rho_{\mathrm{ext}}\in\mathcal{K}_\infty\cup\{0\}$, and matrix $\bar X$ in Definition \ref{Def_1a} associated with $V$ in (\ref{Eq_7a}) are defined as $\alpha(s)=\frac{\lambda_{\min}(\tilde M)}{\lambda_{\max}(C_1^TC_1)}s^2$, $\kappa(s):=(1-\widehat\kappa) s$, $\rho_{\mathrm{ext}}(s):=0$, $\forall s\in\R_{\ge0}$, and $\bar X=\begin{bmatrix}
	\bar X^{11}&\bar X^{12}\\
	\bar X^{21}&\bar X^{22}
	\end{bmatrix}$. Moreover, positive constant $\psi$ in \eqref{Eq_3a} is $\psi=(1+2/\pi)\lambda_{\max}{(\tilde M)}\delta^2$.
\end{proof}

\section{Case Study}

In this section, we apply our results to the temperature regulation of $n\geq 3$ rooms each equipped with a heater and connected on a circle.
The model of this case study is adapted from \cite{meyer} by including stochasticity in the model as additive noise.
The evolution of temperature $T$ sampled at time interval of length $\tau = 9$ minutes
can be described by the interconnected discrete-time stochastic control system
\begin{align}\notag
\Sigma:\left\{\begin{array}{l}{T}(k+1)=A{T}(k)+\gamma T_{h}\nu(k)+ \beta T_{E}+\varsigma(k),\\
y(k)={T}(k),\end{array}\right.
\end{align}
where $A$ is a matrix with diagonal elements $a_{ii}=(1-2\eta-\beta-\gamma\nu_{i}(k))$, $i\in\{1,\ldots,n\}$, off-diagonal elements $a_{i,i+1}=a_{i+1,i}=a_{1,n}=a_{n,1}=\eta$, $i\in \{1,\ldots,n-1\}$, and all other elements are identically zero.
Parameters $\eta$, $\beta$, and $\gamma$ are conduction factors respectively between the rooms $i \pm 1$ and the room $i$,
between the external environment and the room $i$,
and between the heater and the room $i$.
Moreover,  $ T(k)=[T_1(k);\ldots;T_n(k)]$,  $\nu(k)=[\nu_1(k);\ldots;\nu_n(k)]$, $ \varsigma(k)=[\varsigma_1(k);\ldots;\varsigma_n(k)]$, $T_E=[T_{e1};\ldots;T_{en}]$, where $T_i(k)$ and $\nu_i(k)$ are taking values in $[19,21]$ and $[0,0.6]$, respectively, for all $i\in\{1,\ldots,n\}$. The parameter $T_{ei}=-1\,^\circ C$ are the outside temperature $\forall i\in\{1,\ldots,n\}$, and $T_h=50\,^\circ C$ is the heater temperature. Now, by introducing $\Sigma_i$ described as
\begin{align}\notag
\Sigma_i:\left\{\hspace{-2.5mm}\begin{array}{l}T_i(k+1)=(1-2\eta-\beta-\gamma\nu_{i}(k))T_i(k)+\gamma T_h \nu_i(k)+\eta w_i(k)+\beta T_{ei}+\varsigma_i(k),\\
y_{1i}(k)=T_i(k),\\
y_{2i}(k)=T_i(k),\end{array}\right.
\end{align}
one can readily verify that $\Sigma=\mathcal{I}(\Sigma_1,\ldots,\Sigma_N)$ where the coupling matrix $M$ is such that  $m_{i,i+1}=m_{i+1,i}=m_{1,n}=m_{n,1}=1$, $i\in \{1,\ldots,n-1\}$, and all other elements are identically zero. One can also verify that, $\forall i\in\{1,\ldots,n\}$, condition \eqref{Eq_8a} is satisfied with $\tilde M_i=1$, $K_i=0$, $\bar X^{11}=\eta^2(1+\pi_i)$, $\bar X^{22}=-3.38 \eta(1+\pi_i)$, $\bar X^{12}=\bar X^{21}=\eta\lambda $, where $ \lambda = 1-2\eta-\beta-\gamma\nu_{i}(k)$, and selecting some appropriate values for $\widehat\kappa_i$, $\pi_i$. 
Hence, function $V_i(T_i,\hat T_i)=(T_i-\hat T_i)^2$ is an SStF from $\widehat\Sigma_i$ to $\Sigma_i$ satisfying condition \eqref{Eq_2a} with $\alpha_{i}(s)=s^2$ and condition \eqref{Eq_3a} with $\kappa_i(s):=(1-\widehat\kappa_i) s$, $\rho_{\mathrm{iext}}(s)=0$, $\forall s\in\R_{\ge0}$, $\psi_i=(1+2/\pi_i)\delta_i^2$, $G_i=\hat G_i=H_i=I$, and 
\begin{equation}\label{Eq_22}
\bar X_i=\begin{bmatrix} \eta^2(1+\pi_i) & \eta\lambda  \\ \eta\lambda  &  -3.38\eta(1+\pi_i) \end{bmatrix}\!\!,
\end{equation}
where the input $\nu_i$ is given via the interface function in \eqref{Eq_255} as $\nu_i=\hat \nu_i$. 
Now, we look at $\widehat\Sigma=\mathcal{I}(\widehat\Sigma_1,\ldots,\widehat\Sigma_N)$ with a coupling matrix $\hat M$ satisfying condition \eqref{Con_2a} as $\hat M = M$.
Choosing $\mu_1=\cdots=\mu_N=1$ and using $\bar  X_i$ in (\ref{Eq_22}), matrix $\bar X_{cmp}$ in \eqref{Def_3a} reduces to
$$
\bar X_{cmp}=\begin{bmatrix} \eta^2 (1+\pi)I_{n} & \eta\lambda I_{n} \\ \eta\lambda I_{n} &  -3.38\eta (1+\pi)I_{n} \end{bmatrix}\!\!,
$$
and condition \eqref{Con_1a} reduces to
\begin{align}\notag
\begin{bmatrix} M \\ I_n \end{bmatrix}^T\!\!\!\!\!\bar X_{cmp}\!\!\begin{bmatrix} M \\ I_n \end{bmatrix}\!=\!\eta^2(1\!+\!\pi)M^TM+\eta\lambda M+\eta\lambda M^T\!- 3.38\eta(1+\pi)I_{n} \preceq 0,
\end{align}
without requiring any restrictions on the  number or gains of the subsystems.
In order to satisfy the above inequality, we used $M=M^T$\!, and $4\eta^2(1+\pi)+4\eta\lambda-3.38\eta(1+\pi)\preceq 0$ employing Gershgorin circle theorem \cite{bell1965gershgorin} which can be satisfied for the appropriate values of $\eta, \pi$ and $\lambda$. By choosing finite internal input sets $\hat W_{i}$ of $\widehat\Sigma$ such that $\prod_{i=1}^n \hat W_{i} = \hat M\prod_{i=1}^n \hat X_{i},$ condition \eqref{Con111} is also satisfied. 
Now, one can verify that $V(T,\hat T)=\sum_{i=1}^n(T_i-\hat T_i)^2$ is an SSF from  $\widehat\Sigma$ to $\Sigma$ satisfying conditions \eqref{eq:lowerbound2} and \eqref{eq6666}  with $\alpha(s)=s^2$, $\kappa(s):=(1-\widehat\kappa) s$, $\rho_{\mathrm{ext}}(s)=0$, $\forall s\in\R_{\ge0}$, and $\psi=n(1+2/\pi)\delta^2$.

To demonstrate the effectiveness of proposed approach, we fix $n=15$. By taking the state set discretization parameter $\delta = 0.005$, $\widehat \kappa_i = 0.99, \pi_i = 0.04, \eta_i = 0.1, \beta_i=0.022, \gamma_i = 0.05$, $\forall i\in\{1,\ldots,n\}$, one can readily verify that conditions \eqref{Con_1a} and \eqref{Eq_8a} are satisfied. Accordingly, by using the stochastic simulation function $V$ as in inequality (\ref{Eq_25}) and starting the initial states of the interconnected systems $\Sigma$ and $ \widehat \Sigma$ from 20, we guarantee that the distance between outputs of $\Sigma$ and of $\widehat \Sigma$ will not exceed $\varepsilon = 0.63$ during the time horizon $T_d=10$ with probability at least $90\%$, i.e.
\begin{equation*}
\mathbb P(\Vert y_{a\nu}(k)-\hat y_{\hat a \hat\nu}(k)\Vert\le 0.63,\,\, \forall k\in[0,10])\ge 0.9\,\,.
\end{equation*}
Note that for the construction of finite gMDP, we have selected the center of partition sets as representative points. This choice has further tightened the above inequality.

Let us now synthesize a controller for $\Sigma$ via the abstraction $\widehat \Sigma$ such that the controller maintains the temperature of any room in the safe set [19,21]. The idea here is to first design a local control for abstraction $\widehat \Sigma_i$, and then refine it to system $\Sigma_i$ using interface function. Consequently, controller for the interconnected system  $\Sigma$ would be a vector such that each of its components is the controller for the interconnected system $\Sigma_i$. We employ here software tool \texttt{FAUST}  \cite{FAUST15} to synthesize a controller for $\Sigma$ by taking the input set discretization parameter $\theta = 0.04$, and standard deviation of the noise $\sigma_i = 0.28$, $\forall i\in\{1,\ldots,n\}$. A closed-loop state trajectories of the representative room is illustrated in Figure~\ref{Fig2} left. The optimal policy $\nu$ and the associated safety probability for a representative room in the network are plotted in Figure~\ref{Fig3} as a function of initial temperature of the room. The synthesized optimal policy $\nu$ is smoothly decreasing from the maximum input $0.6$ to the minimum $0$ as temperature increases. The maximum safety probability is around the center of the interval $[19,21]$, and its minimums are at the two boundaries. Note that the oscillations appeared in Figure~\ref{Fig3} are due to the state and input discretization.

We now compare the guarantees provided by our approach and by \cite{SAM15}. Note that our result is based on finite gMDP while \cite{SAM15} uses Dynamic Bayesian Network (DBN) to capture the dependencies between subsystems. The comparison is shown in Figure~\ref{Fig6} in logarithmic scale.
In Figure~\ref{Fig6} left, we have fixed $\varepsilon=0.2$ (cf. \eqref{Eq_25}) and plotted the error as a function of discretization parameter $\delta$ and standard deviation of the noise $\sigma$. Our error of \eqref{Eq_25} is independent of $\sigma$ while the error of \cite{SAM15} converges to infinity when $\sigma$ goes to zero. Thus our new approach outperforms \cite{SAM15} for smaller standard deviation of noise.
In Figure~\ref{Fig6} right, we have fixed $\sigma = 0.28$ and plotted the error as a function of discretization parameter $\delta$ and $\varepsilon$.
The error in \cite{SAM15} is independent of $\varepsilon$ while our error increases when $\varepsilon$ goes to zero. Thus there is a tradeoff between $\varepsilon$ and $\delta$ to get better bounds in comparison with \cite{SAM15}.

\begin{figure}
	\centering
	\includegraphics[width=8cm]{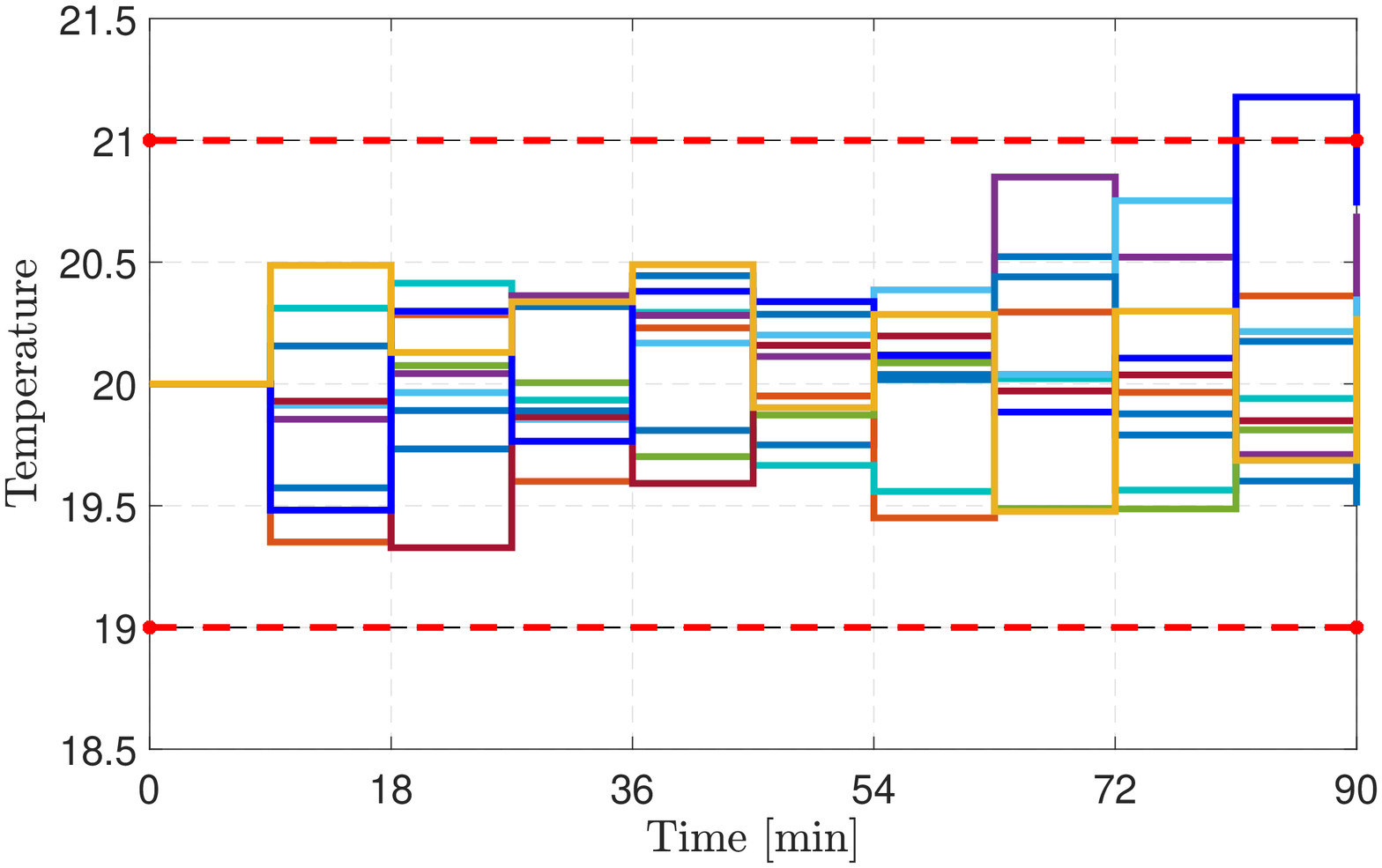}
		\includegraphics[width=8.1cm]{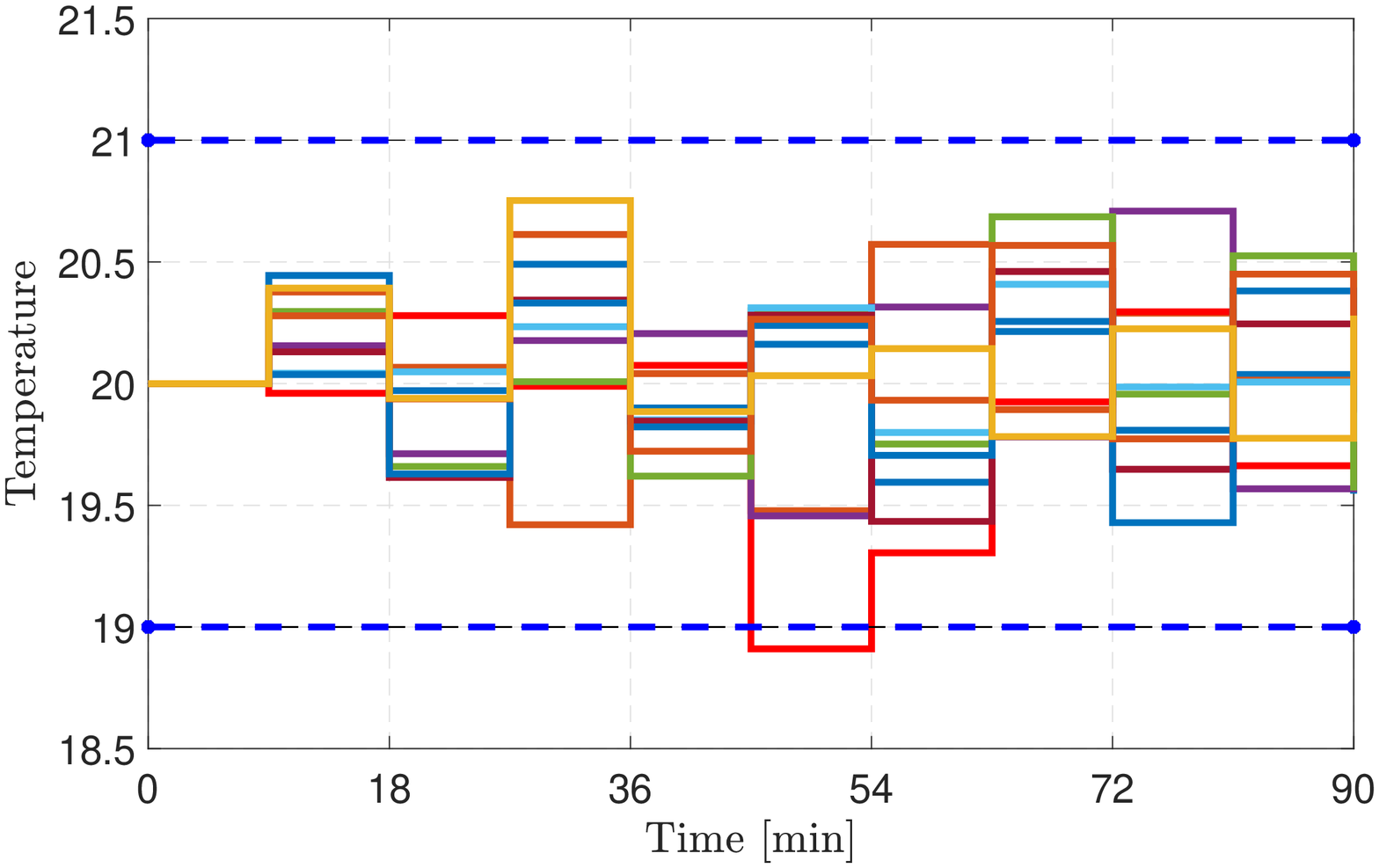}
	\caption{Closed loop state trajectories of a representative room, left plot in a network of 15 rooms and right plot in a network of 200 rooms.}
	\label{Fig2}
\end{figure}
\begin{figure}
	\centering
	\includegraphics[width=8cm]{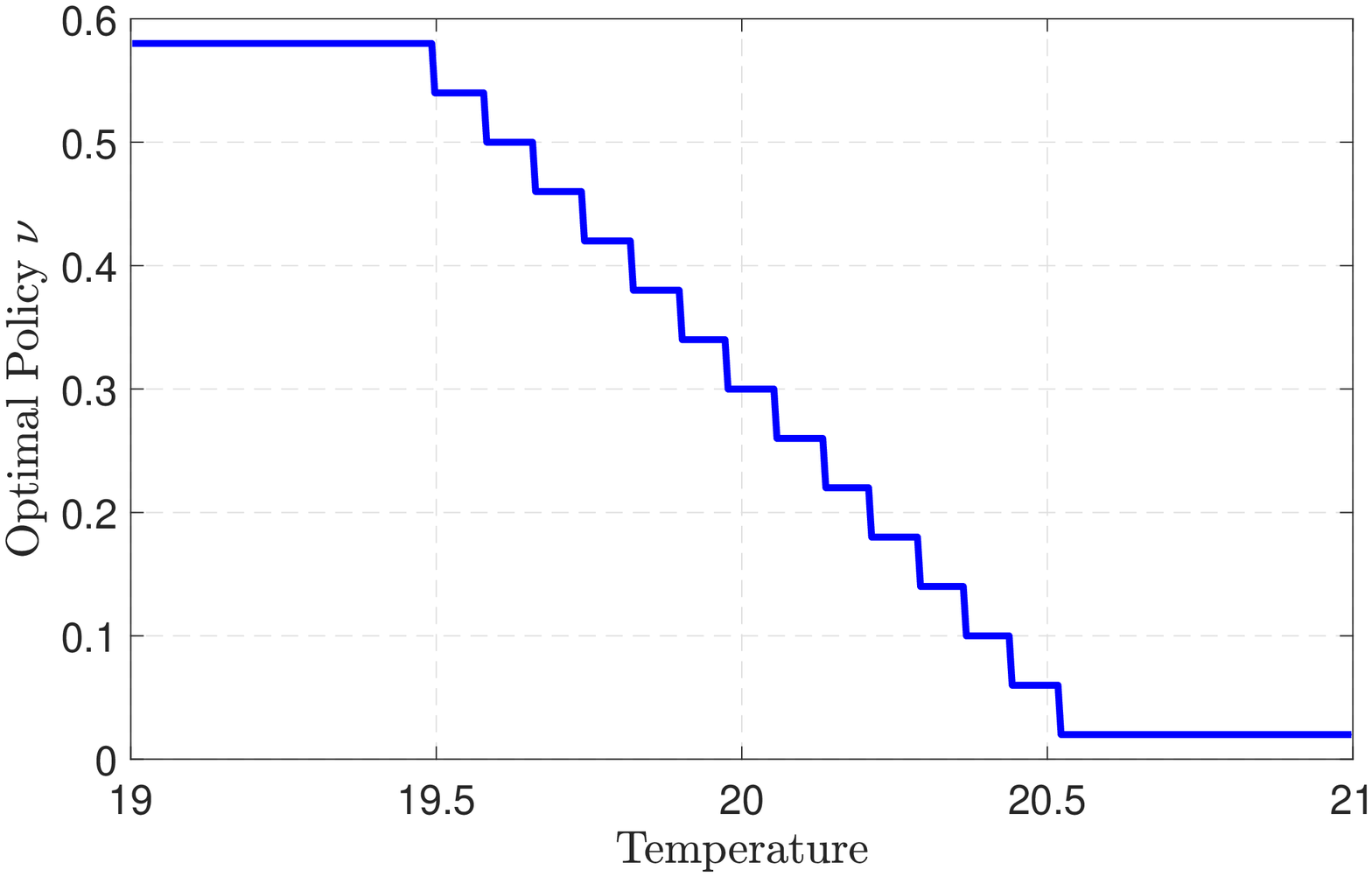}
	\includegraphics[width=8cm]{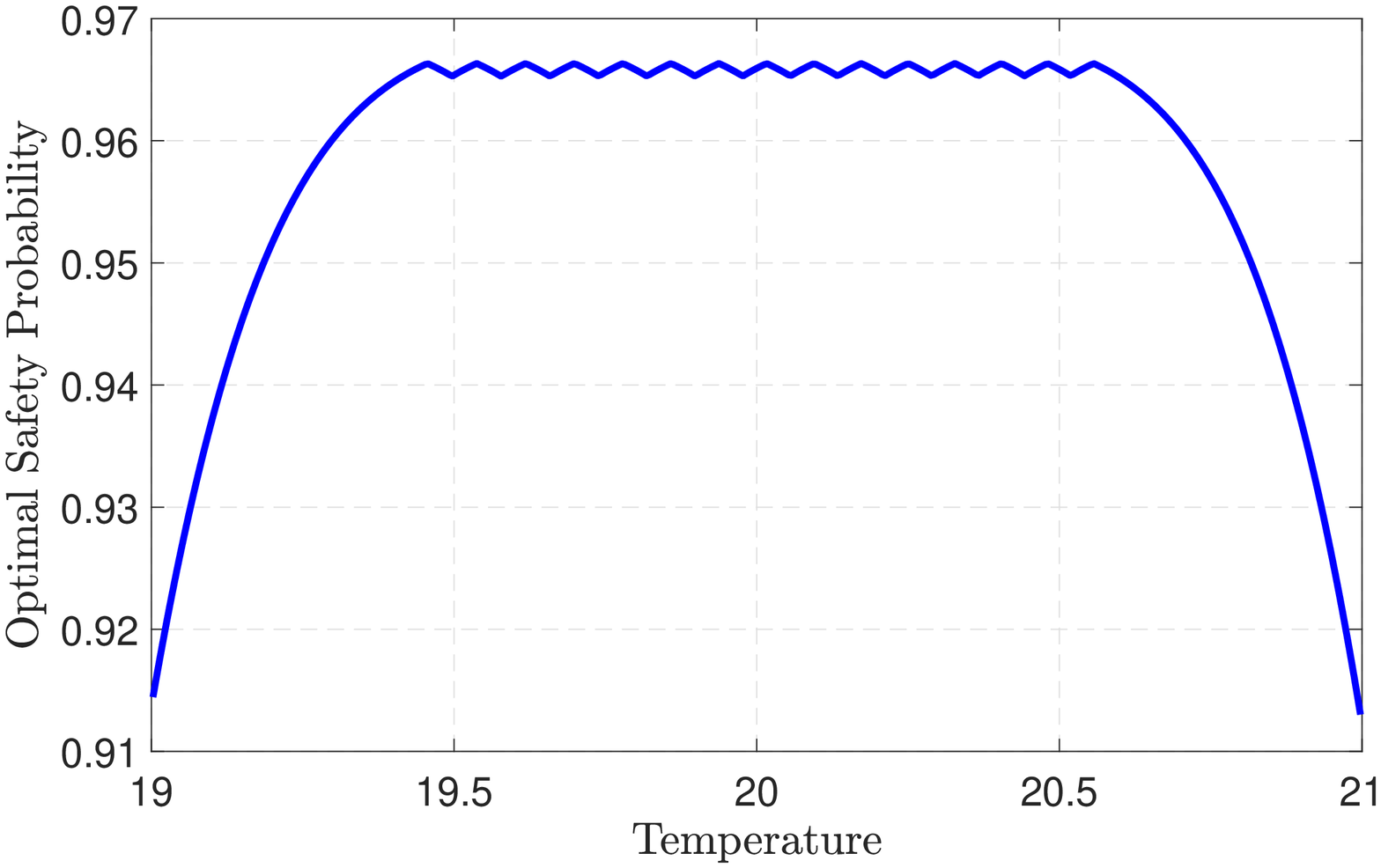}
	\caption{Optimal policy $\nu$ (left) and optimal safety probability (right) with time horizon $T_d = 10$ for a representative room in a network of 15 rooms.}
	\label{Fig3}
\end{figure}

\begin{figure}
	\centering
	\includegraphics[width=8cm]{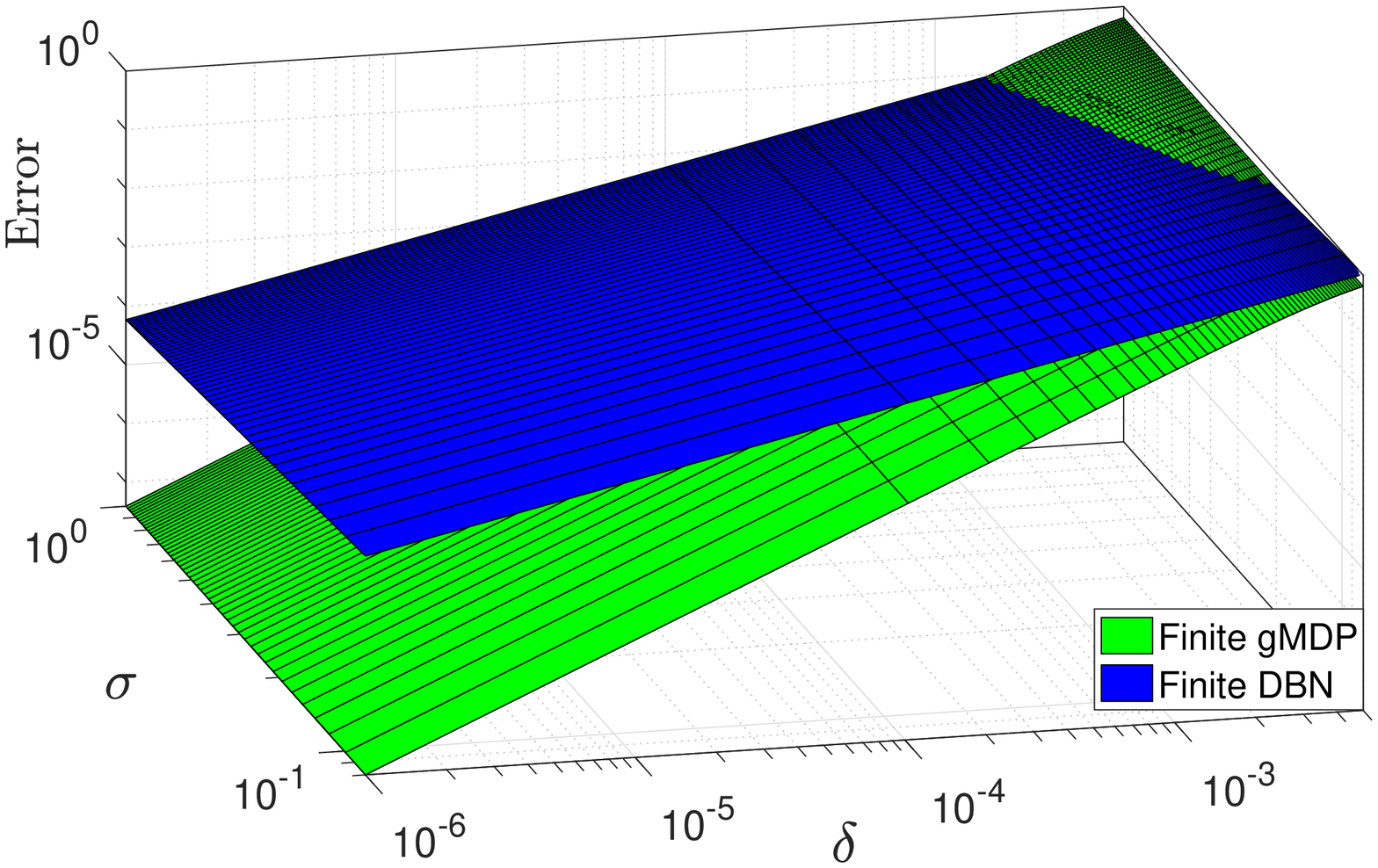}
    \includegraphics[width=8cm]{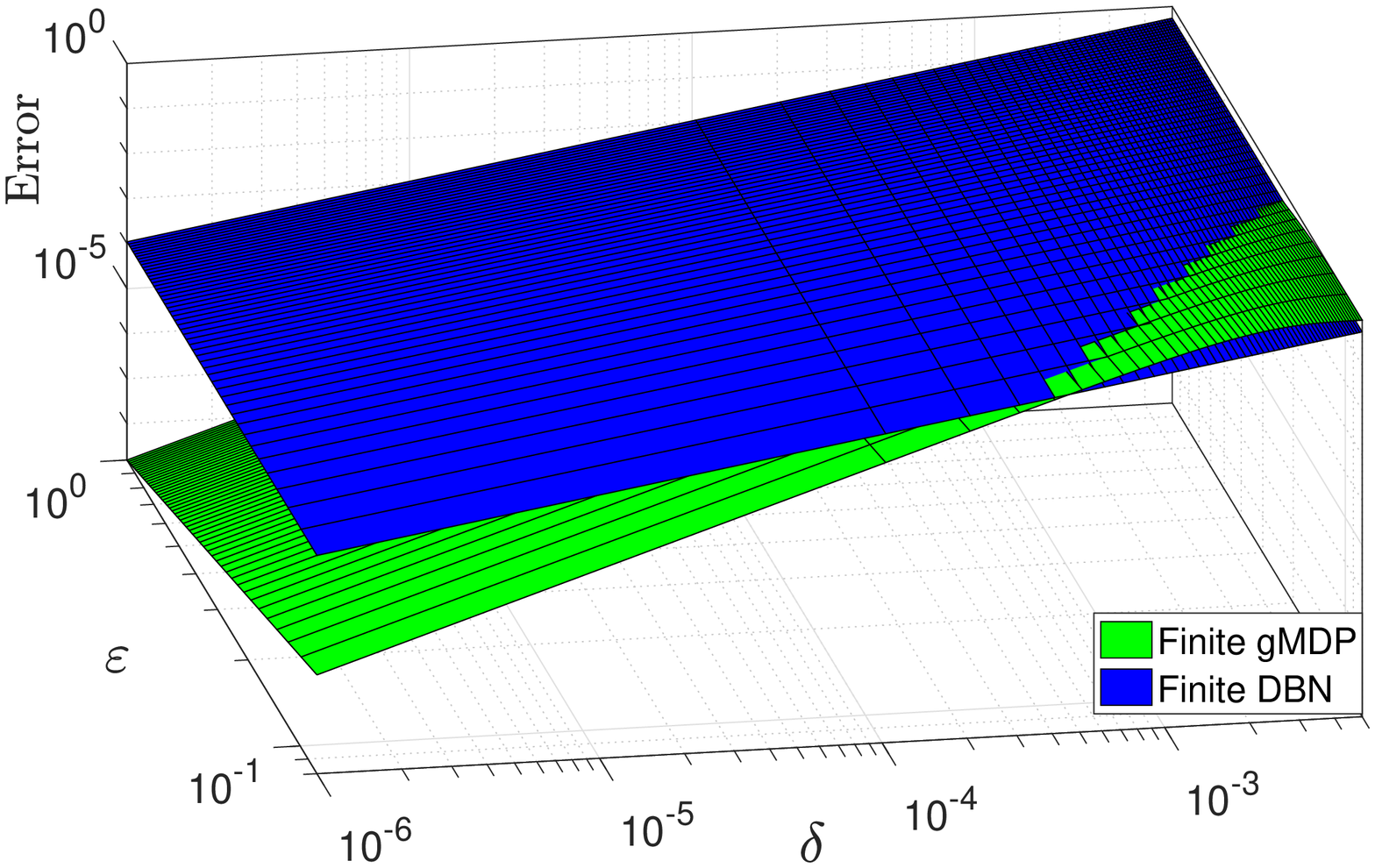}
	\caption{Comparison of error bound provided by the approach of this paper based on finite gMDP with that of [EAM15] based on finite DBN. Plots are in logarithmic scale for a fixed
		$\varepsilon = 0.2$ (cf. (3.5)) in the left and  for a fixed
		noise standard deviation $\sigma = 0.28$ in the right.}
	\label{Fig6}
\end{figure}

In order to show scalability of our approach, we increase the number of rooms to $n=200$. If we take the state set discretization parameter $\delta = 0.005$, $\widehat \kappa_i = 0.99, \pi_i = 0.98, \eta_i = 0.1, \beta_i=0.4, \gamma_i = 0.5$, $\forall i\in\{1,\ldots,n\}$, conditions \eqref{Con_1a} and \eqref{Eq_8a} are readily met. Moreover, if the initial states of the interconnected systems $\Sigma$ and $ \widehat \Sigma$ are started from 20, one can readily verify that the norm of error between outputs of $\Sigma$ and of $ \widehat \Sigma$ will not exceed $0.63$ with probability at least $90\%$ computed by the stochastic simulation function $V$ as in inequality (\ref{Eq_25}) for $T_d=10$. Similarly, we  synthesize a controller for $\Sigma$ via the abstraction $\widehat \Sigma$ by taking the input set discretization parameter $\theta = 0.04$, and $\sigma_i = 0.21$, $\forall i\in\{1,\ldots,n\}$. A closed-loop state trajectories of the representative room is illustrated in Figure~\ref{Fig2} right.

\bibliographystyle{alpha}
\bibliography{biblio}

\end{document}